\documentclass[12pt]{amsart}

\usepackage{amssymb, amsthm, amsfonts, latexsym}
\usepackage[leqno]{amsmath} 
\usepackage{natbib}
\usepackage{multibib}

\usepackage{mathtools}
\usepackage{enumitem}
\usepackage{needspace}
\usepackage{array}
\usepackage{graphicx,booktabs,tabularx,longtable,pdflscape,float}
\usepackage{xurl}
\usepackage[multiple]{footmisc}
\usepackage{comment}

\usepackage[left=1.25in, top=1.25in, right=1.25in, bottom=1.25in]{geometry}
\usepackage{setspace}
\usepackage{tgpagella}
\usepackage{microtype}

\usepackage{xcolor}

\usepackage{etoolbox}
\patchcmd{\section}{\scshape}{\bfseries}{}{}
\makeatletter
\renewcommand{\@secnumfont}{\bfseries}
\makeatother

\theoremstyle{plain}
\newtheorem{theorem}{Theorem}
\newtheorem{corollary}{Corollary}
\newtheorem{lemma}{Lemma}
\newtheorem{proposition}{Proposition}

\theoremstyle{definition}
\newtheorem{definition}{Definition}
\newtheorem{example}{Example}

\newtheoremstyle{boldremark}
 {0.5\topsep}{0.5\topsep}
 {\normalfont}{}
 {\bfseries}{.}
 {.5em}{}
\theoremstyle{boldremark}
\newtheorem{remark}{Remark}

\usepackage[
 colorlinks=true,
 linkcolor=blue,
 citecolor=blue,
 urlcolor=blue,
 filecolor=blue,
 hyperfootnotes=false,
 hyperindex,
 breaklinks
]{hyperref}

\renewcommand{\sectionautorefname}{Section}
\renewcommand{\subsectionautorefname}{Section}
\renewcommand{\subsubsectionautorefname}{Section}

\newcommand{\PR}{\mathrm{PR}}
\newcommand{\RE}{\mathrm{MRU}}
\newcommand{\NW}{\mathrm{NW}}
\newcommand{\Hset}{\mathrm{H}}
\newcommand{\CR}{\mathsf{CR}}
\newcommand{\dis}{\mathrm{dis}}
\newcommand{\exc}{\mathrm{exc}}

\newif\ifsupplement
\supplementtrue  
\newcommand{\suppref}[1]{\ifsupplement\autoref{#1}\else the Supplementary Appendix\fi}
\newcommand{\bodysecref}[1]{\begingroup
  \renewcommand{\sectionautorefname}{Section}%
  \renewcommand{\subsectionautorefname}{Section}\autoref{#1}\endgroup}

\newcommand{\suppsectionref}[1]{\ifsupplement\begingroup
  \renewcommand{\sectionautorefname}{Supplementary Appendix}%
  \renewcommand{\subsectionautorefname}{Supplementary Appendix}%
  \renewcommand{\subsubsectionautorefname}{Supplementary Appendix}%
  \renewcommand{\appendixautorefname}{Supplementary Appendix}%
  \autoref{#1}\endgroup\else the Supplementary Appendix\fi}

\newcommand{\appref}[1]{\begingroup
  \renewcommand{\sectionautorefname}{Appendix}%
  \renewcommand{\subsectionautorefname}{Appendix}%
  \autoref{#1}\endgroup}

\begin{document}

\title[Overlapping Reserves without a Common Priority]{Overlapping Reserves without a Common Priority: Theory and an Application to School Choice in Chile}

\author[Do\u{g}an, Imamura, and Yenmez]{Battal Do\u{g}an \and Kenzo Imamura \and M. Bumin Yenmez}

\thanks{
Do\u{g}an: University of Bristol, Bristol, UK (\texttt{battal.dogan@bristol.ac.uk}); 
Imamura: University of Tokyo, Tokyo, Japan (\texttt{imamurak@e.u-tokyo.ac.jp}); 
Yenmez: Washington University, St. Louis, MO, USA and Durham University Business School, Durham, UK and \"Ozye\u{g}in University, Istanbul, Turkey (\texttt{bumin@wustl.edu}).
We thank participants at SSCW 2026 for helpful comments. This work was partially supported by JSPS KAKENHI Grant Number JP23K12443 and JST ERATO Grant Number JPMJER2301.}

\begin{abstract}

We study matching markets with overlapping reserves governed by heterogeneous priorities. When applicants qualify for multiple reserves, how they are routed can affect reserve utilization and which applicants are admitted. We introduce classify-and-rank choice rules that yield strategy-proof deferred acceptance mechanisms satisfying maximum reserve utilization and priority-respect. We axiomatically characterize our mechanism's admission outcomes for any given classification and reserve order. With two reserve categories, we provide a foundation for the reserve order and establish its distributional implications for access across eligibility groups. Using nine cycles of Chilean school admissions, which feature such
reserves, we document reserve-utilization failures
under the deployed mechanism and show that
their frequency depends sharply on the reserve order. We show how our mechanism can be adapted to Chile's institutional setting to prevent such failures. Under this alternative mechanism, the reserve order generates a distributional pattern across eligibility groups, consistent with our theory.

\vspace{1em}

\end{abstract}

\date{}
\maketitle
\enlargethispage{2pt} 

\section{Introduction}

Assignment systems often reserve seats for different categories of applicants to widen access to education and public employment. Applicants may qualify for several categories, and categories may rank applicants by different criteria.

Chile's centralized school-choice system, the Sistema de Admisi\'on Escolar (SAE), provides a concrete setting.\footnote{For its design and implementation, see \citet{correa2022}; for its development and first decade, see \citet{neilson2026sae}.} At oversubscribed schools, Chile's 2015 School Inclusion Law applies admission priorities in succession, beginning with sibling priority and followed by priority for socioeconomically disadvantaged students up to a 15 percent target. Authorized high-academic-demand schools may also reserve vacancies for high-performing students under a distinct criterion.\footnote{Unused reserved capacity becomes open to all students.} Students may qualify for both reserves, whose priorities differ, so how multi-eligible students are routed determines which reserved opportunities remain for others.

How should seats be allocated when students may be eligible for several categories, each with its own priority order? A mechanism should use reserved seats as fully as possible: no reserve-eligible student who prefers a school to her assignment should be excluded if admitting her, possibly after rerouting the school's assignees across seat types, would increase reserve utilization (\emph{maximum reserve utilization}). It should also respect the priority attached to each seat (\emph{priority-respect}) and give students no incentive to misreport (\emph{strategy-proofness}).

When all seats share one priority ranking, a canonical mechanism achieves these desiderata \citep[e.g.,][]{sonmez/yenmez:22,aygun/turhan:2026}. Among the applicant sets a school can select while maximizing reserve utilization, there is a \emph{Gale-dominant} one, whose members are collectively higher-ranked than those of any other \citep{gale1968optimal,sonmez/yenmez:22}; since every seat type ranks students identically, that choice respects every priority, and using it within deferred acceptance (DA) \citep{gale62} gives a strategy-proof mechanism with maximum reserve utilization and priority-respect.

This logic breaks down once priorities are reserve-specific, so that one student ranks above another for one reserve and below her for another. A natural repair, imposing an auxiliary ranking and seeking a Gale-dominant choice among those that respect priorities and maximize reserve utilization, fails: we construct an
environment with overlapping eligibility and reserve-specific priorities in
which, for every ranking, some school and applicant pool admit no choice
that respects priorities, maximizes reserve utilization, and is
Gale-dominant under that ranking (\autoref{prop_nonexistence_gale}).

We instead introduce a family of mechanisms based on \emph{classify-and-rank} choice rules. A \emph{common-priority classification} partitions the reserves into classes, each containing only categories that share a priority criterion,\footnote{Categories governed by the same criterion may nevertheless belong to different classes.} and a strict reserve order ranks the classes and the categories within each class. The choice rule follows that order: it first selects the students admitted through each class, under the class criterion, and only then assigns each of them to a category, at every step retaining a selection only if it can still be extended to maximum reserve utilization. Embedded in deferred acceptance, these rules are strategy-proof and satisfy maximum reserve utilization and \emph{joint priority-respect}, a strengthening of priority-respect (\autoref{thm:existence}).\footnote{Joint priority-respect lets a student's claim be evaluated after rearranging assignments within her common-priority class, rather than only against the occupant of one seat type (\autoref{sec:joint-priority}).} In the common-priority special case, we recover the characterizations of \citet{sonmez/yenmez:22} as corollaries (\appref{sec:single_priority}).

The grouping itself matters. Processing same-criterion categories one at a time can commit a multi-eligible student to one reserve and thereby exclude a second student in favor of a third who ranks below her under their common priority; selecting through the class first keeps the multi-eligible student's category open. \autoref{ex:grouping-access} shows that no fixed order of individual categories reproduces the grouped rule across applicant pools with capacities, priorities, and eligibility fixed, so the classification, alongside the reserve order, determines whose priority claims are protected.

For a given classification and class order, \emph{class precedence} rules out an alternative admitted set that improves access through a class without displacing a higher-priority student there, disturbing preceding classes, or reducing reserve utilization. With individual rationality, non-wastefulness, priority-respect, maximum reserve utilization, and strategy-proofness, it pins down the school matching (\autoref{thm:characterization}), leaving the classification and class order as policy inputs and the seat-type routing free.

Our construction uses the reserve order to resolve the multiplicity left by
maximum reserve utilization and priority-respect. With two reserve categories, the simplest setting that captures the conflict, we show that one fixed
order governs the reserve-selected set across all applicant pools and
eligibility profiles whenever the elementary conflicts between the two categories' priorities are resolved consistently and two irrelevance
conditions hold (\autoref{ord:thm-two}). This identifies reserve access
without requiring identical category assignments in every problem.

The order also shapes access: placing one reserve category earlier weakly expands reserved-seat access for students eligible only for the other (\autoref{prop:access}). Both orders satisfy the desiderata yet distribute reserved opportunities differently. The comparison holds one school's applicant pool fixed, so it does not rank the two orders for the market as a whole: within deferred acceptance, changing the order also changes which students each school receives, through rejection chains.

Chile is a natural testing ground. \citet{correa2022}, who designed the SAE mechanism, explain that the common-priority approach could not be applied because the reserves use different priorities; their solution, building on \citet{kurata17}, represents reserves as separate subschools through which students are routed sequentially in orders that depend on applicant type. Because students remain ranked at subschools for which they lack eligibility, this can leave reserved capacity underutilized even when rerouting would increase utilization.

We reconstruct the deployed mechanism from public administrative data for the nine admission cycles 2017 to 2025. It reproduces all 3{,}795{,}168 published applicant--cycle assignments exactly and yields the reserve routing, which the public data do not report; statements about routing are conditional on this reconstruction.

Overlapping eligibility and conflicting priorities are common. Across the 127 program--cycles operating both the academic-excellence and the disadvantaged reserve, 35.1 percent of the reserve-eligible applicant--cycle observations listing these programs are eligible for both, and in 121 program--cycles the two reserve priorities cannot be represented by a common ranking. Sequential routing leaves reserved capacity underutilized. We adapt our maximum reserve utilization criterion to Chile's institutional requirements, preserving the seat types of returning students and holders of special-needs (PIE) positions and keeping siblings admitted. Under the deployed reserve order, 25 of the 127 program--cycles fail this criterion, with a combined shortfall of 285 reserve positions, 264 of them at program--cycles where no rejected student could have been admitted to raise reserve utilization: most of the shortfall therefore concerns how admitted students are routed. Reversing the order of the two reserves, which the legislation does not prescribe, raises these to 104 program--cycles and 962 positions, with failures in every cycle through 2025, even though the overlapping programs fall to three from 2023 with the end of the transitional academic regime. Under this order, program--cycles where rejected eligible students have claims account for 591 positions, an upper bound on the part of the shortfall that concerns admission (\autoref{sec:res_uti_deployed}).

Sibling-adjusted deferred acceptance (SA-DA), our adaptation of classify-and-rank to Chile's sibling priority, secured enrollment, and special-needs admissions, is strategy-proof within each grade, taking sibling priorities as given, and satisfies the adapted criterion under either reserve order. The academic-first version changes the programs assigned to only 89 applicant--cycle observations. Both versions retain returning students' deployed routing orders at their current programs, whereas the reverse-order sequential-routing counterfactual reverses them as well (\autoref{sec:sibling-adjusted-empirical}).

Finally, the order shapes who gains access. Between academic-first and disadvantaged-first SA-DA, 1{,}185 applicant--cycle observations receive different programs. Among affected students eligible for only one reserve category, more than 81 percent in each group prefer the order that places the other category first, consistent with our access result (\autoref{prop:access}), although the empirical comparison also incorporates rejection chains and Chile's institutional features (\autoref{sec:sibling-adjusted-empirical}).

The framework applies beyond Chile. In Indian state-level professional admissions, applicants may qualify for several reservation categories under different criteria: many categories share the general academic ranking, sports and National Cadet Corps (NCC) quotas use achievement-based criteria, and in some states defence quotas use service-based priorities. Kerala's 2025 engineering-admission lists show the overlap: 104 of 181 sports/NCC applicants also appear in a community or economically weaker section (EWS) category, and 387 appear in both the ex-servicemen and a community or EWS list. Unlike in Chile, where the two reserves use distinct criteria, categories sharing a criterion can be administered jointly, which is where common-priority classification applies. Related structures arise through grade-based and aptitude-test routes in Sweden, grade-based and institution-specific routes in Denmark, first-diploma and regular routes in Norway, and aptitude-based and general admission routes in some English secondary schools.\footnote{For the Indian systems, see \citet{tnea2026brochure,kerala2026categories,apeapcet2025special}. The Kerala counts are calculated from the 2025 engineering-admission lists \citep{kerala2025categories}. For the European examples, see \citet{uhr2026selection,antagning2024allocation} for Sweden, \citet{ug2026admission,ug2026quota2} for Denmark, \citet{norway2022admissions} for Norway, and \citet{england2021schoolcode} for England. These systems differ in their allocation procedures and treatment of unused capacity.}

\textbf{Related Literature.} Our common-priority benchmark is the meritorious horizontal choice rule of \citet{sonmez/yenmez:22}.\footnote{\citet{sonmezyenmez2024} use this rule to analyze matching mechanisms.} Classify-and-rank recovers its selected set when all seat types share one criterion and all reserves form one class (\autoref{cor:single-priority}). Related work studies overlapping types and stability \citep{kurata17}, transfers of unfilled reserves \citep{aygtur20,aytur2023}, hierarchical reserves \citep{aygun/turhan:2026}, and distributional preferences \citep{echmekyen26}.

Prior work develops allocation rules with category-specific priorities and reserve orders. \citet{pathak/sonmez/unver/yenmez:2024}, \citet{aziz2024efficient}, \citet{evren2023reserve}, and \citet{banerjee/eichhorn/kempe:2023} study identical-unit allocation with category-specific priorities. \citet{delacretaz21} allows such priorities in a simultaneous system that can divide an applicant's allocation across categories. Closest to our construction, \citet{sun2025fair} proposes the sequential category updating rule, which processes categories in a fixed order and, within each category, agents in priority order, committing a pair only while the commitments remain extendable to a maximum assignment. All of these papers consider a single institution with agents who care only about receiving a unit, whereas we study matching markets with multiple institutions and preferences over them.

When every common-priority class is a singleton, the reserve component of classify-and-rank coincides with the rule of \citet{sun2025fair} (\autoref{rem:scu}); class precedence adapts his respect for precedence over categories, and the school-level identification in \autoref{thm:characterization} parallels his characterization of the rule under a strict precedence order. The open-seat stages differ: our reserves are soft, so seats not used by reserve-eligible students revert to open competition, whereas in his model each category, including the open category, has a fixed capacity and unused reserved units remain unallocated, which is why maximum cardinality is a separate axiom there and is implied by non-wastefulness here. Beyond this special case, we embed the choice rules in deferred acceptance and establish strategy-proofness; introduce common-priority classes, which select students through same-criterion categories jointly rather than one category at a time, with joint priority-respect across those categories; provide the behavioral foundation for two reserve categories (\autoref{ord:thm-two}) and the access comparison (\autoref{prop:access}); and evaluate the mechanisms empirically. \citet{cavallo/dogan} establish substitutability for related coordinated choice rules.

\citet{kominers_sonmez} allow slot-specific priorities in matching markets; our construction preserves extendability to maximum reserve utilization at every commitment, which sequential slot filling need not do. \citet{westkamp13} studies sequential components with different criteria but without overlapping reserves. Our access comparison relates to \citet{dur/kominers/pathak/sonmez:2018,dur/pathak/sonmez:2020}, who study reserved seats relative to open or merit seats; the latter also characterize orders relative to socioeconomic targeting objectives. We compare single-eligible students' reserve access when two overlapping reserves with potentially different priorities reverse order, holding the pool fixed.

Our characterization uses the lifting method of \citet{Dogan/Imamura/Yenmez:2025}.\footnote{Earlier DA characterizations include \citet{alcalde1994top}, \citet{kojman10}, and \citet{ehlers/klaus:2016}; \citet{fernandez2020regret} shows that DA is the unique quantile-stable mechanism that is regret-free truth-telling, a weakening of strategy-proofness.} Outcomes here also specify seat types, as in \citet{hu/zhang:2026} on Chinese high-school admissions.

The broader reserve literature, initiated by \citet{hayeyi13} as an alternative to majority quotas \citep{kojima_impossibilities} and extended to soft bounds by \citet{ehayeyi14}, studies axiomatic foundations \citep[e.g.,][]{echyen12,dogan2017comment,dogan_yildiz_2022,atila_aram_2024,imamura25}, flexible targets \citep[e.g.,][]{kamada_kojima_aer,kamada_kojima_jet,kamada_kojima_te}, distributional constraints \citep[e.g.,][]{GOTO201640,fratro:2017}, and endogenous quotas \citep{kumano/kurino:2022}.

For Chile's SAE, \citet{correa2022} describe the mechanism they helped design and implement; \citet{neilson2026sae}, who advised the Ministry on its design and development, reviews its first decade; and \citet{escobar/huerta:2026} study the socioeconomic reserve's size. We study the interaction of overlapping reserves with different priorities, their utilization, and the distributional effects of their order.
\section{Model}\label{sec:model}

\subsection{Setup}
Consider an environment with finite sets of students $\mathcal{I}$, schools
$\mathcal{S}$, and reserve categories $\mathcal{T}$. Each student
$i\in\mathcal{I}$ is eligible for a set of reserve categories
$\tau(i)\subseteq\mathcal{T}$, where $\tau:\mathcal{I}\to 2^{\mathcal{T}}$ is
the \textbf{reserve eligibility function}.\footnote{Eligibility may also be
school-specific: write $\tau_s(i)$ and use the corresponding eligibility sets
in each school's choice problem. The constructions and results apply with
this school index; we suppress it in the theoretical notation.} For instance, in our Chilean
application, $\mathcal{T}=\{\text{academic excellence},
\text{disadvantaged}\}$,\footnote{The SAE's special-needs (PIE) places follow a separate procedure
(\autoref{sec:empirical}) and are not included in $\mathcal{T}$.} and a high-achieving student $i$ from a
low-income household would have
$\tau(i)=\{\text{academic excellence},\text{disadvantaged}\}$, whereas a
low-income student $j$ who is not high-achieving would have
$\tau(j)=\{\text{disadvantaged}\}$. We abbreviate the two categories as $\exc$ and $\dis$.

Each student $i \in \mathcal{I}$ has a strict \textbf{preference ranking} $P_i$
over $\mathcal{S} \cup \{\varnothing\}$, where $\varnothing$ represents being
unmatched.\footnote{We define all rankings as strict linear orders over
their domains.} A school $s\in \mathcal{S}$ is \textbf{acceptable} to
student $i\in \mathcal{I}$ if $s \mathrel{P_i} \varnothing$. For each $s,s' \in
\mathcal{S} \cup \{\varnothing\}$, we write $s \mathrel{R_i} s'$ to denote weak
preference, that is, $s \mathrel{P_i} s'$ or $s=s'$. Since $P_i$ and $R_i$
determine each other, we use whichever is more convenient and let $R=(R_i)_{i\in
\mathcal{I}}$ denote a \textbf{preference profile}. For each student
$i\in\mathcal{I}$, we write $R_{-i}\coloneqq(R_j)_{j\in\mathcal{I}\setminus\{i\}}$ for
the preference profile of all students other than $i$.

Each school $s\in\mathcal{S}$ has a capacity of $q_s\in\mathbb{Z}_+$ seats.\footnote{We 
use $\mathbb{Z}_+=\{0,1,2,\ldots\}$.} For each reserve category $t \in \mathcal{T}$, school $s$ reserves
$r_s^t\in \mathbb{Z}_+$ seats for students eligible for category $t$, where
$\sum_{t \in \mathcal{T}} r_s^t \le q_s$.
Seats not allocated through reserves are \textbf{open seats}, for which all
students are eligible.

For $I\subseteq\mathcal{I}$ and $t\in\mathcal{T}$, let
$I^{t}\coloneqq\{i\in I:t\in\tau(i)\}$ be the students in $I$ eligible for
$t$, and let $I^o\coloneqq I$; the elements of $\mathcal{T}\cup\{o\}$ are
\textbf{seat types}.

For each school $s \in \mathcal{S}$, seat types are administered according to
institutional \textbf{admission criteria}, such as a lottery, a
merit score, or a category-specific assessment. A criterion map $\gamma_s$ records the criterion $\gamma_s(t)$ governing each
seat type $t\in\mathcal{T}\cup\{o\}$. Each criterion $c$ is represented by a
strict \textbf{priority ranking} $\succ_s^c$ over $\mathcal{I}$, and the
priority ranking of a seat type is that of its criterion,
$\succ_s^t\coloneqq\succ_s^{\gamma_s(t)}$. Two seat types share a criterion
precisely when $\gamma_s$ assigns them the same value, whatever the
realized rankings; distinct criteria may rank students differently, making
priorities reserve-specific. The collection of these rankings over all
schools and criteria is the \textbf{priority profile}.

For each school $s\in\mathcal S$, a \textbf{common-priority classification}
$\mathcal B_s$ is a partition of the reserve categories $\mathcal T$ such
that $\gamma_s$ is constant on every class $B\in\mathcal B_s$.
Categories under distinct criteria cannot share a class even if their
rankings coincide, whereas categories sharing a criterion may be split
across classes; the classification records which categories are
administered jointly. Open seats are not in $\mathcal B_s$, though they may
share a criterion with a class.

In the main matching model, all primitives other than student preferences are fixed. 
A \textbf{problem} is therefore a preference profile $R$.

\subsection{Matchings, Choice Rules, and Mechanisms}
A \textbf{matching} is a pair $(\mu,\sigma)$ that assigns each student to a
school or leaves her unmatched, and gives each admitted student a seat type,
reserved or open. Here $\mu$ is a school matching and $\sigma =
(\sigma_s)_{s\in\mathcal{S}}$ is a profile of reserve matchings, one
for each school. Students care only about their assigned school; the seat type 
records the category through which the student is admitted.

Formally, a \textbf{school matching} is a function
$\mu:\mathcal{I}\to \mathcal{S}\cup\{\varnothing\}$ such that, for each school
$s\in \mathcal{S}$, $|\mu^{-1}(s)| \le q_s$, where $\mu^{-1}(s)$ denotes the set
of students assigned to $s$.
Given a school matching $\mu$, for each school $s\in\mathcal{S}$, a
\textbf{reserve matching} for $\mu$ at $s$ is a function $\sigma_s :
\mu^{-1}(s)\to \mathcal{T}\cup\{o\}$ that satisfies \textbf{reserve
feasibility}: for each reserve category $t\in\mathcal{T}$, $\sigma_s^{-1}(t)
\subseteq \mathcal{I}^t$ and $|\sigma_s^{-1}(t)| \le r_s^t$.

Open-seat assignments are bounded only by the school's capacity, so
reserves are \textbf{soft}: unused reserved capacity is used for open-seat
admissions, with those students assigned seat type $o$.

Let $\mathcal{M}$ denote the set of all matchings. A pair $(I_s,\sigma_s)$ is a
\textbf{matching for school} $s$ if $I_s \subseteq \mathcal{I}$, $|I_s|\leq
q_s$, and $\sigma_s : I_s \to \mathcal{T} \cup \{o\}$ satisfies reserve
feasibility.

A \textbf{mechanism} $\varphi$ maps each problem $R$ to a matching
$\varphi(R)\in\mathcal{M}$; let $\varphi_i(R)\in\mathcal{S}\cup\{\varnothing\}$
denote the school to which $\varphi(R)$ assigns student $i\in\mathcal{I}$.

For a school $s\in\mathcal S$, a \textbf{choice rule} $C_s$ assigns to each
applicant set $I\subseteq\mathcal I$ a matching
$C_s(I)=\bigl(\overline C_s(I),\sigma_s^{C}(I)\bigr)$ for school $s$ with
$\overline C_s(I)\subseteq I$. We refer to $\overline C_s$ as the
\textbf{restricted choice rule} induced by $C_s$; it records the selected students
but not their seat types.

\section{Desiderata for Matchings and Mechanisms}\label{sec:desiderata}

We first present four axioms for matchings---individual rationality,
non-wastefulness, priority-respect, and maximum reserve utilization---and
strategy-proofness for mechanisms.\footnote{A mechanism satisfies a matching
property if it selects a matching satisfying that property at every problem.}
We then introduce our new axiom, \emph{joint priority-respect}.

\subsection{Basic Axioms}\label{sec:basic-axioms}

\begin{definition}\label{def:individual_rationality}
A matching $(\mu,\sigma)$ is \textbf{individually rational} for problem $R$ if,
for each student $i\in\mathcal I$,
$\mu(i)\mathrel{R_i}\varnothing$.
\end{definition}

Individual rationality requires that no student prefers the outside option to
her assigned school.

\begin{definition}\label{def:non_wastefulness}
A matching $(\mu,\sigma)$ is \textbf{non-wasteful} for problem $R$ if, for each
student $i\in\mathcal I$ and school $s\in\mathcal S$,
$s\mathrel{P_i}\mu(i)$ implies $|\mu^{-1}(s)|=q_s$.
\end{definition}

Non-wastefulness requires that any school preferred by a student to her
assignment be operating at full capacity. Since open seats are capped only by total capacity, reserve feasibility
never prevents filling a vacant seat.

\begin{definition}\label{def:priorities}
A matching $(\mu,\sigma)$ is \textbf{priority-respecting} for problem $R$ if,
for any two students $i,j\in\mathcal I$, school $s\in\mathcal S$, and seat type
$t\in\mathcal T\cup\{o\}$ such that $\mu(i)=s$ and $\sigma_s(i)=t$,
\[
s\mathrel{P_j}\mu(j)
\quad\Longrightarrow\quad
j\notin\mathcal I^t
\quad\text{or}\quad
i\succ_s^t j.
\]
\end{definition}

In words, a student who prefers $s$ to her assignment has no justified claim
to a seat at $s$: for every student assigned to $s$, she is either ineligible
for that student's seat type or ranked below that student under its priority.

Our next axiom requires that reserved-seat opportunities not be left unused
when they could benefit eligible students.\footnote{This reserve-utilization
objective is central to the approach of \citet{sonmez/yenmez:22}.}
For a matching $(I_s,\sigma_s)$ for $s$, let
$\rho_s(I_s,\sigma_s)\coloneqq \sum_{t\in\mathcal T}|\sigma_s^{-1}(t)|$
denote the number of reserved seats it utilizes. A student $i$ is
\textbf{reserve-eligible at $s$} if $t\in\tau(i)$ for some reserve category
$t$ with $r_s^t>0$.

\begin{definition}\label{def:maximum_reserve_utilization}
A matching $(\mu,\sigma)$ satisfies \textbf{maximum reserve utilization} for
problem $R$ if there exist no school $s$, reserve-eligible student $i$ with
$s\mathrel{P_i}\mu(i)$, and matching $(I_s',\sigma_s')$ for $s$ such that
$i\in I_s' \subseteq \mu^{-1}(s)\cup\{i\}$
and
\[
\rho_s(I_s',\sigma_s')
>
\rho_s\bigl(\mu^{-1}(s),\sigma_s\bigr).
\]
\end{definition}

Thus no reserve-eligible student who prefers a school to her assignment can
be admitted there, possibly after rejecting some admitted students and
rerouting the rest, in a way that increases reserve utilization; the
alternative must be feasible but need not respect priorities. The axiom is an access requirement, but because
it compares alternatives with the realized utilization
$\rho_s(\mu^{-1}(s),\sigma_s)$ rather than with the utilization the admitted
students could attain, it also restricts routing: whenever some
reserve-eligible student prefers $s$ to her assignment, it requires that the
students admitted to $s$ be routed so as to fill the maximum number of
reserved seats they can jointly fill, and that admitting her could not raise
that maximum (\autoref{lem:mru-demand} in \appref{supp:characterization}).
At a school that no reserve-eligible student prefers to her own assignment,
the axiom does not restrict routing.

The four axioms do not determine a unique matching, even in a
one-school setting.

\begin{example}\label{ex:multiplicity}
Consider one school with capacity two and two reserve categories
$\mathcal T=\{\dis,\exc\}$, with one reserved seat in each. There are three
students, with
\[
\tau(i_1)=\{\dis,\exc\},\qquad
\tau(i_2)=\{\dis\},\qquad
\tau(i_3)=\{\exc\},
\]
and priorities
\[
i_1\succ_s^{\dis}i_2\succ_s^{\dis}i_3,
\qquad
i_1\succ_s^{\exc}i_3\succ_s^{\exc}i_2.
\]
The school is acceptable to every student, and the open-seat priority is
irrelevant. Exactly two matchings satisfy all four axioms:
$i_1\to\dis,\ i_3\to\exc$ and $i_1\to\exc,\ i_2\to\dis$. Both are
individually rational, non-wasteful, and fill every reserved seat. Both
respect priorities: in the first, $i_2$ cannot displace $i_3$ through
$\exc$, for which she is ineligible, or $i_1$ through $\dis$, who outranks
her; the second is symmetric.
\end{example}

\autoref{sec:results} and \autoref{sec:interpreting-order} resolve this
multiplicity. The next axiom is stated for mechanisms rather than matchings.

\begin{definition}\label{def:strategy_proofness}
A mechanism $\varphi$ is \textbf{strategy-proof} if, for each problem $R$,
student $i\in\mathcal I$, and preference ranking $R_i'$,
$\varphi_i(R)\mathrel{R_i}\varphi_i(R_i',R_{-i})$.
\end{definition}

Thus, truthful reporting is a dominant strategy for every student.

\subsection{Joint Priority-Respect}\label{sec:joint-priority}

Priority-respect evaluates claims against the particular seat types assigned.
The following example shows why this may be insufficient when reserve
categories share a criterion.

\begin{example}\label{ex:joint}
Retain the setting of \autoref{ex:multiplicity}, but suppose the two reserve
categories share the same admission criterion, with priority
$i_1\succ_s i_2\succ_s i_3$,
and form a single common-priority class,
$\mathcal B_s=\{\{\dis,\exc\}\}$. Consider the two matchings
\[
i_1\to\dis,\quad i_3\to\exc,
\qquad\text{and}\qquad
i_1\to\exc,\quad i_2\to\dis.
\]
Both satisfy the four axioms; the first respects priorities because $i_2$
is ineligible for $\exc$ and ranks below $i_1$ for $\dis$. Yet it excludes
$i_2$ in favor of the lower-priority $i_3$, while reassigning $i_1$ to
$\exc$ would let $i_2$ replace $i_3$ with both reserved seats still filled.
With a shared criterion, admitting $i_3$ over the higher-ranked $i_2$ is
hard to justify.
\end{example}

This example motivates \emph{joint priority-respect}, which strengthens
priority-respect by allowing seat assignments within a common-priority class
to be rearranged when evaluating priority claims.

Fix a school $s\in\mathcal{S}$ and a nonempty group of seat types
$G\subseteq\mathcal T\cup\{o\}$. The seat types in $G$ \emph{share a
criterion} if $\gamma_s$ is constant on $G$; we then write $\succ_s^G$ for
their common priority. For a matching $M=(\mu,\sigma)$, let
$A_s^G(M)\coloneqq\{i\in\mu^{-1}(s):\sigma_s(i)\in G\}$ be the students
assigned to $s$ through seat types in $G$, and write
$A_s^t(M)\coloneqq A_s^{\{t\}}(M)$. For $I\subseteq\mathcal I$, let $n_s^G(I)$
be the largest number of students in $I$ who can simultaneously be assigned
reserved seats of categories in $G\cap\mathcal T$, each to a category for
which she is eligible and within the reserve capacities. Unlike $\rho_s$,
which counts realized reserve assignments, $n_s^G(I)$ measures attainable
reserve utilization, and it never counts open seats, even when $o\in G$.
Write $n_s(I)\coloneqq n_s^{\mathcal T}(I)$.

\begin{definition}\label{def:joint-priority}
A matching $M=(\mu,\sigma)$ satisfies \textbf{joint priority-respect} for
problem $R$ if it satisfies priority-respect and, for each school
$s\in\mathcal S$ and each common-priority class $B\in\mathcal B_s$, the
following condition holds. Let $G\coloneqq B\cup\{o\}$ if open seats are
governed by the criterion of $B$, and $G\coloneqq B$ otherwise. For every
student $i\in\mathcal I$ with $s\mathrel{P_i}\mu(i)$ and every
$j\in A_s^G(M)$ with $i\succ_s^G j$,
\[
n_s^G\!\left(
\bigl(A_s^G(M)\setminus\{j\}\bigr)\cup\{i\}
\right)
<
n_s^G\!\left(A_s^G(M)\right).
\]
\end{definition}

In words, within a class, a higher-priority claimant may be passed over
only if replacing a lower-priority admitted student would reduce the
maximum achievable reserve utilization; when the open seat shares the
class criterion, open-seat assignments may be rearranged as well. In
\autoref{ex:joint}, the first matching satisfies priority-respect but
violates the joint condition, whereas separating the two categories into singleton classes permits it, since the comparison then no longer reroutes
$i_1$ and $i_2$ jointly. The classification thus determines the scope of
the fairness comparison. In a single-priority environment with
$\mathcal B_s=\{\mathcal T\}$, the joint condition is the no-justified-envy
condition of \citet{sonmez/yenmez:22}, with priority-respect governing the
routing separately (\appref{sec:single_priority}); when every class is a
singleton sharing no criterion with open seats, it adds nothing to
priority-respect.

\section{Achieving the Desiderata}\label{sec:results}
The desiderata are jointly achievable but do not determine a unique
matching. We resolve the multiplicity with the common-priority
classification and a strict order over reserve categories, define the
resulting school choice rules, use them within deferred acceptance, and
characterize the resulting mechanism.

\subsection{Classify-and-Rank Choice Rules}

Fix a school $s\in\mathcal S$. Its choice rule is built from the
classification $\mathcal B_s$ and a strict order $\triangleright_s$ over the
reserve categories, the \textbf{reserve order}, that lists the categories of
each class consecutively; we call the pair $(\mathcal B_s,\triangleright_s)$ a
\textbf{classify-and-rank specification}. Index the classes as
$B_s^1,\ldots,B_s^{K_s}$ so that $t\triangleright_s t'$ whenever $t\in B_s^k$,
$t'\in B_s^\ell$, and $k<\ell$. The reserve order thus fixes a \textbf{class
order} and a \textbf{within-class order}; the first determines which students
are selected, the second only how they are routed within their class. The
specification is fixed across preference reports, unlike Chile's
applicant-specific subschool orders.

We now specify the choice $\CR_s(I)$ from an applicant set
$I\subseteq\mathcal I$.\footnote{We suppress the fixed specification from the
notation.} Let $E_s(I)\coloneqq\{(t,i)\in\mathcal T\times I: t\in\tau(i),\
r_s^t>0\}$ be the feasible reserve--student pairs. A \textbf{reserve
assignment} is a set $X\subseteq E_s(I)$ in which each student appears at
most once and no reserve capacity is exceeded; $\mathcal X_s(I)$ denotes the
set of reserve assignments, so $n_s(I)=\max_{X\in\mathcal X_s(I)}|X|$. The
rule proceeds in three stages: it commits students to classes, then to
categories within their classes, retaining each commitment only if it still
extends to a reserve assignment of size $n_s(I)$, and finally fills the
remaining capacity by open-seat priority.

\smallskip
\noindent\textbf{Stage~1: Class commitment.}
Each $X\in\mathcal X_s(I)$ induces the class commitment
$\pi_s(X)\coloneqq\{(k,i):(t,i)\in X\text{ for some }t\in B_s^k\}$. Consider
the pairs $(k,i)$ such that $i$ is eligible for a category of $B_s^k$ with
reserved seats, in class order and, within a class, in the order of the
class priority. Starting from the empty set, retain $(k,i)$ if the retained
pairs together with $(k,i)$ are contained in $\pi_s(X)$ for some
$X\in\mathcal X_s(I)$ with $|X|=n_s(I)$. Let $Z_s(I)$ be the retained set; by
\autoref{lem:lexicographic_dominance} (\appref{app:cnr-properties}), it is
realizable with $|Z_s(I)|=n_s(I)$.

\smallskip
\noindent\textbf{Stage~2: Exact reserve assignment.}
Conditional on $Z_s(I)$, consider the pairs in $E_s(I)$ in reserve order
and, within a category, in the order of its priority. Starting from the
empty set, retain $(t,i)$ if the retained pairs together with $(t,i)$ are
contained in some $X\in\mathcal X_s(I)$ with $\pi_s(X)=Z_s(I)$. By the same
lemma, the resulting reserve assignment $X_s(I)$ is unique and realizes
$Z_s(I)$.

\smallskip
\noindent\textbf{Stage~3: Open seats.}
Let $\overline{\CR}_s(I)$ consist of the students assigned by $X_s(I)$
together with the highest-$\succ_s^o$-priority students among the remaining
members of $I$, up to the remaining capacity $q_s-|X_s(I)|$. The
classify-and-rank choice is
$\CR_s(I)=\bigl(\overline{\CR}_s(I),\sigma_s^{\CR}(I)\bigr)$, where
$\sigma_s^{\CR}(I)(i)=t$ if $(t,i)\in X_s(I)$ and $\sigma_s^{\CR}(I)(i)=o$
for every other selected student.

The rule can be computed in polynomial time by implementing the extendability
tests through maximum-weight matching (see \suppsectionref{supp:computation}).

\begin{remark}\label{rem:scu}
When every common-priority class is a singleton, the reserve component of the
classify-and-rank choice rule coincides with the sequential category updating
rule of \citet{sun2025fair} in the corresponding reserve setting. Both process
categories in the reserve order and students within each by its priority,
retaining a pair only while the commitments remain extendable to maximum
utilization. The open-seat stage differs, since our reserves are soft and his
category capacities are hard.
\end{remark}

\begin{remark}
Simple greedy slot filling can strand a later reserve even when all reserved
seats could jointly be filled. The extendability test rules this out,
distinguishing classify-and-rank from greedy slot filling under slot-specific
priorities \citep{kominers_sonmez}.
\end{remark}

In \autoref{ex:multiplicity}, the two class orders select the two matchings
left unresolved there: $\dis\triangleright_s\exc$ gives $i_1\to\dis$ and
$i_3\to\exc$, whereas $\exc\triangleright_s\dis$ gives $i_1\to\exc$ and
$i_2\to\dis$, so reversing the order reroutes the multi-eligible student and
changes which single-eligible student is selected.

\begin{example}\label{ex:classes-rule}\label{ex:grouping-access}
A school has capacity three and one place in each of categories
$t_1,t_2,t_3$; the students are $i_1,\ldots,i_6$. Categories $t_1$ and $t_2$
share an admission criterion, with priority
$i_1\succ i_2\succ\cdots\succ i_6$, whereas $t_3$ ranks $i_6$ first. Eligibility is
\[
\tau(i_1)=\{t_1,t_2,t_3\},\quad \tau(i_2)=\tau(i_5)=\{t_1\},\quad
\tau(i_3)=\tau(i_4)=\{t_2\},\quad \tau(i_6)=\{t_2,t_3\}.
\]
Take the classes $\{t_1,t_2\}$ and $\{t_3\}$, with reserve order
$t_1\triangleright_s t_2\triangleright_s t_3$.

From pool $\{i_1,i_2,i_3,i_6\}$, classify-and-rank selects $i_1,i_2$ through the first
class and $i_6$ through $t_3$, routing $i_2$ to $t_1$ and $i_1$ to $t_2$.
Classify-and-rank with singleton classes and $t_1$ before $t_2$ instead commits
$i_1$ to $t_1$, admits $i_3$ through $t_2$, and selects $\{i_1,i_3,i_6\}$. Both outcomes fill
every reserve and respect
seat-specific priorities. Only the grouped outcome protects $i_2$'s claim
against the lower-ranked $i_3$ across their common-priority class.

Putting $t_2$ before $t_1$ recovers the grouped outcome in that pool, but fails
in pool $\{i_1,i_4,i_5,i_6\}$. There, grouping selects $\{i_1,i_4,i_6\}$, whereas
$t_2$-first commits $i_1$ to $t_2$ and selects $\{i_1,i_5,i_6\}$. Since $i_6$ receives $t_3$ in either pool whatever $t_3$'s position, no fixed
order of the three singleton categories reproduces the grouped rule in both
pools, although capacities, priorities, and eligibility are unchanged. Under
the classification $(\{t_1,t_2\},\{t_3\})$, both grouped outcomes satisfy joint
priority-respect, whereas each singleton order violates it in one pool:
grouping governs access beyond the choice of order.
\end{example}

Whether categories sharing a criterion should be grouped is a policy
choice. Grouping treats them as one pool: seats are interchangeable up to
eligibility and the common priority governs who is admitted, which is what
protects $i_2$'s claim in the example. Separating them gives each category its
own constituency, so the category, and not only the shared criterion,
determines who is admitted; this is natural when categories serve different
groups under one ranking.

\subsection{Deferred Acceptance Based on Classify-and-Rank}

Fix a classify-and-rank specification at every school, and write
$\CR\coloneqq(\CR_s)_{s\in\mathcal S}$ and
$\overline{\CR}\coloneqq(\overline{\CR}_s)_{s\in\mathcal S}$ for the profiles of
choice rules and restricted choice rules. In the \textbf{DA mechanism based
on} $\CR$, a version of deferred acceptance (DA) \citep{gale62}, students
apply one at a time: a student not currently held by any school applies to
her most-preferred acceptable school that has not yet rejected her. That
school applies $\CR_s$ to the students it holds together with the new
applicant, holds the chosen students in their assigned seat types, and
rejects the rest. Since each student applies to each school at most once,
the process ends with final school and seat-type assignments.

\begin{theorem}\label{thm:existence}
For any profile of classify-and-rank specifications, the DA mechanism based on
$\CR$ satisfies individual rationality, non-wastefulness, joint
priority-respect, maximum reserve utilization, and strategy-proofness. In
particular, it also satisfies priority-respect.
\end{theorem}

The proof in \appref{supp:existence} establishes substitutability of
$\overline{\CR}_s$ using a fixed assignment valuation whose weights encode
reserve utilization, class commitments, category assignments, and open-seat
priority. Size monotonicity follows because the rule selects
$\min\{q_s,|I|\}$ students from $I$. Standard cumulative-offer results
\citep{hatfi04} and equivalence to DA yield strategy-proofness; individual
rationality follows from the procedure. Each school's terminal choice,
including seat types, equals $\CR_s$ applied to all its proposers
(\autoref{cor:da_co_seats}), so non-wastefulness, maximum reserve utilization,
and joint priority-respect carry over to the matching.

\subsection{Characterization}
\label{sec:characterization}

\autoref{thm:existence} shows that DA based on classify-and-rank satisfies
the desiderata. We now show that, for a given classification and class
order, it is the only mechanism that does so and also respects class
precedence, an axiom that formalizes the implications of the
classification and class order for school access.
\autoref{sec:interpreting-order} then asks where the order comes from.

For a matching $M=(\mu,\sigma)$, or a matching $M'=(I_s',\sigma_s')$ for $s$, let
$A_s^k(\cdot)$ denote the students assigned at $s$ through the categories of
class $B_s^k$.
At problem $R$, such a matching is an \textbf{alternative matching} for $s$ if every $i\in I_s'$ satisfies $s\mathrel{R_i}\mu(i)$.

\begin{definition}
\label{def:class_precedence}
Fix a classify-and-rank specification
$(\mathcal B_s,\triangleright_s)$ for each school $s\in\mathcal S$. A matching
$M=(\mu,\sigma)$ \textbf{respects class precedence} at school
$s\in\mathcal S$ for problem $R$ if there exist no common-priority class $B_s^k$,
student $i\notin A_s^k(M)$, and alternative matching
$M'=(I_s',\sigma_s')$ for $s$ such that
\begin{enumerate}[label=(\roman*)]
\item $M'$ uses at least as many reserved seats at $s$ as $M$;
\item $A_s^\ell(M')=A_s^\ell(M)$ for every $\ell<k$;
\item every $j\in A_s^k(M)$ with $j\succ_s^{B_s^k}i$ belongs to
$A_s^k(M')$;
\item $i\in A_s^k(M')$; and
\item $I_s'\neq\mu^{-1}(s)$.
\end{enumerate}
\end{definition}

Conditions (i)--(iv) fix the students selected by earlier classes, though
not their internal routing, and permit assigning a student to the current
class without displacing a higher-priority student there or reducing
utilization; lower-priority students and later classes may be rearranged,
which with multiple categories may require coordinated reassignments
rather than a one-for-one exchange. Condition~(v) restricts the axiom to changes in the admitted set. Without it,
the axiom would also rule out improvements that only reroute students;
classify-and-rank satisfies this stronger version, but the characterization
does not need it.

\Needspace{9\baselineskip}
\begin{theorem}
\label{thm:characterization}
Fix a classify-and-rank specification
$(\mathcal B_s,\triangleright_s)$ for each school $s\in\mathcal S$. The DA
mechanism based on $\CR$ satisfies individual rationality, non-wastefulness,
priority-respect, maximum reserve utilization, strategy-proofness, and class
precedence at every school. Conversely, every mechanism satisfying these
properties induces the same school matching as the DA mechanism based on
$\CR$.
\end{theorem}

None of the characterizing properties is implied by the others (see
\suppsectionref{supp:independence}). The proof works school by school. \autoref{lem:access_choice} shows that a
non-wasteful, priority-respecting, utilization-maximizing matching for $s$
that respects class precedence selects the classify-and-rank set.
When every reserve-eligible student is already admitted, the utilization
axiom does not pin down routing, and the proof uses the nesting of
class-selected sets as reserve cardinality increases
(\autoref{lem:nested-class-access}). The lifting theorem of
\citet{Dogan/Imamura/Yenmez:2025} then extends school-level uniqueness to
the mechanism. The classification and class order remain policy inputs, and
joint priority-respect, though not among the characterizing properties,
also holds (\autoref{thm:existence}).

\appref{sec:single_priority} relates the construction to the common-priority
benchmark. With a single priority ranking, classify-and-rank recovers the
meritorious horizontal choice rule of \citet{sonmez/yenmez:22}, and their
characterizations follow (\autoref{cor:single-priority}); with reserve-specific priorities, whatever ranking is fixed as the reference,
some school and applicant pool admit no Gale-dominant choice
(\autoref{prop_nonexistence_gale}).

\section{Understanding the Reserve Order}
\label{sec:interpreting-order}

\autoref{sec:results} takes the reserve order as given. This section asks
where it comes from and what it does, in the simplest setting that captures
the conflict: a single school with two reserve categories, as in the
Chilean application. The main result is that a rule
which resolves the simplest conflict the same way every time, with one
student eligible for both categories competing against one student
eligible for each, must select as a fixed reserve order would in every
problem (\autoref{ord:thm-two}); the order in turn shapes who gains access
(\autoref{prop:access}). Both results concern the reserve stage, so we
consider a school without open seats.\footnote{With open seats,
\autoref{prop:access} still holds for classify-and-rank, which fills the
reserves first, but it compares reserved-seat access rather than admission,
which the order can reverse by changing who competes for the open seats
(see \suppsectionref{ord:boundary}). \autoref{ord:thm-two} applies to any
rule whose reserve stage satisfies the axioms below; with open seats, the
axioms of \autoref{sec:desiderata} do not imply them, because a student
admitted through an open seat has no claim against the holders of reserved
seats.}

\subsection{Behavioral Foundation}
\label{ord:main}
Call the two categories $t_1$ and $t_2$; each forms a singleton common-priority class, with $r^1,r^2>0$ reserved seats and priority rankings $\succ^1$ and $\succ^2$, and in general we index objects attached to $t_k$ by $k$.

A \textbf{reserve problem} is a pair $(I,\tau)$, where
$I\subseteq\mathcal I$ is an applicant set and
$\tau:I\to 2^{\{t_1,t_2\}}$
specifies students' reserve eligibility. A student may be eligible for neither reserve, only $t_1$, only $t_2$, or both. We consider all such problems, with capacities and the two priority rankings fixed. Unlike in the main model, eligibility varies across problems, so that the axioms below can compare problems that differ only in eligibility.

A \textbf{reserve rule} $\Psi$ assigns to every reserve problem
$(I,\tau)$ a reserve assignment $X=\Psi(I,\tau)$, defined as in
\autoref{sec:results} with eligibility $\tau$.
For a reserve assignment $X$, write $X_t\coloneqq \{i:(t,i)\in X\}$ for the students assigned through $t$, and $A(X)\coloneqq\{i:(t,i)\in X\text{ for some }t\}$ for its assigned students; students in $I\setminus A(X)$ are unassigned.

The rule \textbf{maximizes reserve utilization} if, at every reserve
problem, it assigns the maximum feasible number of students through reserved
seats. It satisfies \textbf{priority-respect} if every student assigned
through a reserve has higher priority for that reserve than every unassigned student eligible for it: $i\in I\setminus A(X)$, $t\in\tau(i)$, and $j\in X_t$ imply $j\succ^t i$.

\begin{definition}
A reserve rule satisfies \textbf{irrelevance of rejected students} if deleting an unassigned student leaves the entire reserve
assignment unchanged.
\end{definition}

\begin{definition}
A reserve rule satisfies \textbf{irrelevance of unused eligibility} if deleting an eligibility link $(t,i)$ that is not used in the chosen
reserve assignment leaves the entire reserve assignment unchanged.
\end{definition}

Deleting an unused link leaves the chosen assignment feasible, and the axiom
requires the rule to keep choosing it. The axiom is a consistency condition across problems that differ
only in eligibility,
distinct from incentive compatibility for eligibility reports
\citep{aybo16}.

\begin{definition}
An \textbf{elementary conflict} is a reserve problem containing
students $i_1,j,i_2$ such that
\[
\tau(i_1)=\{t_1\},
\qquad
\tau(j)=\{t_1,t_2\},
\qquad
\tau(i_2)=\{t_2\},
\]
and
$j\succ^1 i_1$ and $j\succ^2 i_2$.
In addition, there are $r^1-1$ students eligible only for $t_1$, each ranked
above $i_1$ under $\succ^1$, and $r^2-1$ students eligible only for $t_2$, each
ranked above $i_2$ under $\succ^2$. There are no other students or eligibility
links.
\end{definition}

Exactly two assignments satisfy maximum reserve utilization and
priority-respect in an elementary conflict. Both assign the other $r^1-1$ and $r^2-1$ single-eligible students through their own categories. One assigns $j$ through $t_1$ and $i_2$ through $t_2$;
the other assigns $i_1$ through $t_1$ and $j$ through $t_2$. The first follows
$t_1$'s priority in assigning $j$ rather than $i_1$ through $t_1$, leaving $t_2$'s
remaining seat for $i_2$; the second follows $t_2$'s priority symmetrically.

\begin{definition}
\label{ord:def-conflict}
A reserve rule \textbf{resolves elementary conflicts consistently} if $j$ is assigned through the same reserve category in every elementary
conflict.
\end{definition}

Consistent resolution requires the same category's priority to govern
regardless of the students involved. It does not impose an order directly
on arbitrary pools.

\Needspace{8\baselineskip}
\begin{theorem}
\label{ord:thm-two}
Suppose a reserve rule $\Psi$ maximizes reserve utilization, satisfies
priority-respect, irrelevance of rejected students, and irrelevance of
unused eligibility, and resolves elementary conflicts consistently.
Then there exists an order $\triangleright$, either $t_1\triangleright t_2$ or $t_2\triangleright t_1$, such that, for every
reserve problem $(I,\tau)$,
\[
A\bigl(\Psi(I,\tau)\bigr)
=
A\bigl(X^{\triangleright}(I,\tau)\bigr),
\]
where $X^{\triangleright}(I,\tau)$ is the reserve assignment selected by classify-and-rank under order $\triangleright$ (Stages~1 and~2 of \autoref{sec:results}). Moreover, whenever elementary
conflicts arise, they are resolved according to the priority of the first category in $\triangleright$.
\end{theorem}

Conversely, both fixed-order classify-and-rank rules satisfy the assumptions:
deleting an unassigned student or an unused link preserves the
lexicographically optimal maximum reserve assignment, and elementary
conflicts are resolved through the first category. The conditions
therefore pin down exactly the two fixed-order selections. The result
concerns selected sets, not category assignments; it does not say which
order is preferable, and it leaves the order undetermined when no
elementary conflict arises. The proof reduces any disagreement in selected
sets to an elementary conflict resolved against the proposed order
(\appref{ord:proof-two}).

\subsection{Access Implications}

Fix $(I,\tau)$, let $I^1$ and $I^2$ be the students in $I$ eligible for $t_1$
and $t_2$ under $\tau$, and write $X^{1\triangleright2}$ and $X^{2\triangleright1}$
for $X^{\triangleright}(I,\tau)$ under the orders $t_1\triangleright t_2$ and
$t_2\triangleright t_1$.

\begin{proposition}
\label{prop:access}
Giving one reserve category precedence weakly expands reserve access for students who are eligible only for the other category. That is,
$$
X_2^{2\triangleright1}\setminus I^1
\subseteq
X_2^{1\triangleright2}\setminus I^1,
\qquad
X_1^{1\triangleright2}\setminus I^2
\subseteq
X_1^{2\triangleright1}\setminus I^2,
$$
where $X_2^{2\triangleright1}\setminus I^1$, for instance, consists of the
students assigned through $t_2$ under $t_2\triangleright t_1$ who are not
eligible for $t_1$.

\end{proposition}

The proof is in \appref{supp:access}. It shows more generally that every student assigned through a category when that category is processed first remains reserve-selected when it is processed second, possibly through the other category: $X_2^{2\triangleright1}\subseteq A(X^{1\triangleright2})$ and $X_1^{1\triangleright2}\subseteq A(X^{2\triangleright1})$.

The comparison holds the applicant pool and all institutional
primitives fixed. It need not extend to individual welfare
comparisons within deferred acceptance, where rejection chains change the
applicant pools.

With multi-category classes, the order still determines whose priority prevails, but the access inclusion in \autoref{prop:access} can fail, because routing within a class can overturn it (see \suppsectionref{ord:boundary}).

\section{Application to School Choice in Chile}\label{sec:empirical}

Chile's centralized School Admission System (\emph{Sistema de Admisi\'on Escolar}, or SAE) provides a natural setting for studying overlapping reserves with different priorities. Some students qualify for both the academic-excellence and the disadvantaged reserve, $\exc$ and $\dis$, which use different priority criteria.

The law fixes who is eligible for each reserve and how each reserve ranks
students, but not how a student eligible for both is routed between them;
that is a feature of the implementation. After describing these requirements and reconstructing the deployed mechanism using data from nine admission cycles (2017--2025), we investigate failures of reserve utilization under sequential routing. We then develop sibling-adjusted deferred acceptance (SA-DA), which adapts classify-and-rank to Chile's institutional setting and provably prevents such failures, and measure how much it changes assignments and how the reserve order distributes access.

\subsection{Admission Objectives and the Deployed Mechanism}
\label{sec:empirical-setting}\label{sec:chile-exposure}

The SAE assigns places in publicly funded schools through centralized applications: families rank schools, places are allocated subject to capacities and admission priorities, and returning students retain secured enrollment at their current school if they do not obtain a preferred alternative. At oversubscribed schools, the School Inclusion Law (Law No.~20,845) applies priorities in succession: siblings, socioeconomically disadvantaged students up to a 15 percent target, children of school employees, and former students; the disadvantaged capacity reported for new admission is the shortfall remaining after enrolled disadvantaged students are counted. Authorized high-academic-demand schools may also allocate places by academic excellence under a different criterion, and schools operating a \emph{Programa de Integraci\'on Escolar} (PIE) may designate places for students with permanent special educational needs under a separate procedure, with unused PIE capacity available for general admission. \suppsectionref{supp:empirical-implementation} gives the full institutional account and the statutory references.

The SAE currently uses a version of the student-proposing deferred acceptance mechanism, processed grade by grade from higher to lower grades. Assignments in a higher grade can generate sibling priority in lower grades. Families applying for multiple children can also choose a joint application option, under which the school assigned to an older sibling is moved to the top of the younger siblings' preference lists. 

At each step of deferred acceptance, each \textbf{program}, a school--grade combination at which capacities and reserves are reported, decides which students to hold. In the deployed mechanism \citep{correa2022}, building on \citet{kurata17}, this decision is made using internal \textbf{subschools} corresponding to different seat types: regular seats, disadvantaged ($\dis$) reserves, academic-excellence ($\exc$) reserves at authorized programs, and special-needs (PIE) positions where applicable. Each student proposes to a program's subschools in a \textbf{routing order} that depends on her characteristics and enrollment status. For non-returning students, $\exc$ precedes $\dis$ whenever both are available. In general, $\dis$ precedes regular seats, except for non-disadvantaged siblings, who are considered for regular seats before $\dis$. Non-disadvantaged non-siblings therefore also propose to $\dis$ before regular seats and can hold $\dis$ positions. Returning students follow a separate routing procedure at their current program that protects their secured enrollment. Together, the subschool priorities and applicant-specific routing orders determine the program's choice rule, specifying which students it retains and through which subschools.

Siblings matter in particular. Following the statutory precedence, the
deployed $\dis$ subschool ranks them ahead of ordinary disadvantaged
students, so non-disadvantaged siblings, though routed to regular seats
first, can still take $\dis$ positions when regular seats run out, reducing
the capacity available to disadvantaged students.

\subsection{Data and Reconstruction}
\label{sec:data}

We use publicly available administrative data from Chile's Ministry of Education covering nine SAE admission cycles from 2017 to 2025 \citep{mineduc_sae_data}. The data contain students' applications and ranked school lists, program capacities and reserves, eligibility and priority information, lottery outcomes, family links, and final assignments. Although the data report each applicant's assignment status, school, and grade, they do not identify the subschool through which the applicant was admitted.

Using these data and the deployed mechanism described by \citet{correa2022}, we reconstruct the assignment process, including secured enrollment, PIE admissions, grade-by-grade processing, and family-related adjustments. Our reconstruction exactly reproduces the published assignment status, school, and grade for all 3{,}795{,}168 applicant--cycle observations (see \suppref{tab:replication}). Seat types come from the reconstructed routing; since routing is not observed, exact replication of assignments does not establish that this routing is unique (see \suppsectionref{supp:empirical-implementation}).

Our counterfactuals hold fixed submitted applications, capacities, eligibility, priority criteria, lottery realizations, and family links, while updating sibling priorities and family-related preference adjustments to reflect the assignments generated in higher grades.

\subsection{Reserve Utilization under the Deployed Mechanism}
\label{sec:res_uti_deployed}

An \textbf{overlapping-reserve program} is a program with positive capacity in both $\exc$ and $\dis$ reserves. Across 127 such program--cycles, 46{,}978 applicant--cycle observations list at least one overlapping-reserve program and are eligible for at least one of its reserves; 16{,}504 (35.1 percent) are eligible for both. Moreover, at 121 of these 127 program--cycles, the two reserve priorities cannot be represented by a common ranking of students eligible for both (see \suppref{tab:exposure}).

To ask whether sequential routing causes avoidable failures of reserve utilization, we adapt the maximum reserve utilization (MRU) criterion to Chile's institutional setting: siblings must remain admitted, although not necessarily through $\dis$, and secured enrollment and PIE are protected.

Fix an outcome, either the deployed assignment or a counterfactual one, and a
program $s$, and let $I_s$ and $\sigma_s$ denote its assignees at $s$ and their
seat assignments. An assignee occupying a reserve subschool for which she is not eligible is treated as an open-seat assignee, so that $\rho_s(I_s,\sigma_s)$ counts eligible reserve holders. Let $I_s^{\mathrm{prot}}\subseteq I_s$ denote the returning students and the holders of PIE positions, let $I_s^{\mathrm{sib}}\subseteq I_s$ denote the assignees with sibling priority at $s$, and let $I_s^{\mathrm{dem}}$ denote the reserve-eligible students who prefer $s$ to their assignments, all under that outcome.

\begin{definition}\label{def:immediate-violation}
An outcome satisfies \textbf{protection-respecting maximum reserve utilization} at program $s$ if there exists no matching $(I_s',\sigma_s')$ for $s$ such that
$I_s'\subseteq I_s\cup I_s^{\mathrm{dem}}$,
$I_s^{\mathrm{prot}}\cup I_s^{\mathrm{sib}}\subseteq I_s'$,
$\sigma_s'(j)=\sigma_s(j)$ for every $j\in I_s^{\mathrm{prot}}$,
$\sigma_s'(j)\neq \exc$ for every $j\in I_s^{\mathrm{sib}}$ with $\sigma_s(j)\neq \exc$, and
\[
\rho_s(I_s',\sigma_s')>\rho_s(I_s,\sigma_s).
\]
The \textbf{shortfall} at $s$ is the largest value of $\rho_s(I_s',\sigma_s')-\rho_s(I_s,\sigma_s)$ over such matchings. A student $i\in I_s^{\mathrm{dem}}$ has a \textbf{protection-respecting MRU claim} at $s$ if such a matching exists with $i\in I_s'\subseteq I_s\cup\{i\}$.
\end{definition}

\autoref{tab:immediate-violations} evaluates this criterion at the
overlapping-reserve programs under the deployed and reverse routings, and
\autoref{prop:sibling-adjusted-outcome} applies it to SA-DA. Compared with
\autoref{def:maximum_reserve_utilization}, the alternative matching must keep every
protected assignee admitted: returning and PIE students retain their seat
types, and siblings remain admitted, possibly through a different seat type
but not through an $\exc$ position they do not already hold. Unlike
\autoref{def:maximum_reserve_utilization}, which adds a single claimant, the
alternative may admit several eligible students or none, so the shortfall
also captures failures that arise from rerouting the assignees alone. The
restriction on siblings reflects how sibling precedence operates in the SAE: it applies
at the disadvantaged subschool and at regular seats (see
\suppsectionref{supp:empirical-implementation}), whereas $\exc$ positions are
allocated by academic priority. The protection therefore guarantees a
sibling's admission but not access to the academic reserve; moving her into
an $\exc$ position would let a protected admission occupy a position
allocated by academic priority, possibly ahead of rejected students with
higher academic priority. Dropping the restriction would only enlarge the set
of alternatives, so the shortfalls in \autoref{tab:immediate-violations}
would weakly increase; the proof that SA-DA satisfies the criterion
(\autoref{prop:sibling-adjusted-outcome}) also relies on it. 

We also examine how these failures depend on the reserve order through a \textbf{reverse-order} counterfactual that exchanges the positions of $\exc$ and $\dis$ in students' within-program routing orders, leaving all other institutional procedures unchanged. Because the deployed routing orders of returning students at their current program place $\dis$ before $\exc$, the exchange places $\exc$ before $\dis$ for them while placing $\dis$ before $\exc$ for non-returning students (see \suppsectionref{supp:empirical-implementation}). It is an applicant-specific reversal of the deployed routing rather than a uniform disadvantaged-first order.

\begin{table}[htbp]
\centering
\caption{Protection-respecting MRU failures under sequential routing}
\label{tab:immediate-violations}
\small
\setlength{\tabcolsep}{3.5pt}
\begin{tabular}{lrrrrrrrrr}
\toprule
& & \multicolumn{4}{c}{Deployed routing}
& \multicolumn{4}{c}{Reverse order} \\
\cmidrule(lr){3-6}\cmidrule(lr){7-10}
Cycle & Programs
& Failing & Shortfall & \begin{tabular}[b]{@{}r@{}}No\\claim\end{tabular} & \begin{tabular}[b]{@{}r@{}}Claim-\\ants\end{tabular}
& Failing & Shortfall & \begin{tabular}[b]{@{}r@{}}No\\claim\end{tabular} & \begin{tabular}[b]{@{}r@{}}Claim-\\ants\end{tabular} \\
\midrule
2017 & 6 & 0 & 0 & 0 & 0 & 2 & 12 & 0 & 37 \\
2018 & 15 & 2 & 21 & 21 & 0 & 10 & 106 & 28 & 1{,}643 \\
2019 & 35 & 8 & 109 & 93 & 58 & 24 & 277 & 112 & 2{,}779 \\
2020 & 23 & 9 & 91 & 86 & 5 & 22 & 196 & 38 & 1{,}849 \\
2021 & 23 & 4 & 21 & 21 & 0 & 21 & 205 & 132 & 1{,}630 \\
2022 & 16 & 2 & 43 & 43 & 0 & 16 & 132 & 61 & 1{,}160 \\
2023 & 3 & 0 & 0 & 0 & 0 & 3 & 8 & 0 & 639 \\
2024 & 3 & 0 & 0 & 0 & 0 & 3 & 18 & 0 & 565 \\
2025 & 3 & 0 & 0 & 0 & 0 & 3 & 8 & 0 & 515 \\
\midrule
Total & 127 & 25 & 285 & 264 & 63 & 104 & 962 & 371 & 10{,}817 \\
\bottomrule
\end{tabular}

\begin{minipage}{0.94\textwidth}
\footnotesize
\emph{Notes:} Overlapping-reserve programs. ``Failing'' counts
program--cycles with a positive shortfall, and ``Shortfall'' sums
the shortfalls in reserve positions (\autoref{def:immediate-violation}); ``No claim'' is the part of the shortfall at program--cycles where no rejected student has a claim, so that it can be recovered only by rerouting admitted students. ``Claimants'' counts
applicant--cycle observations with at least one claim; claimants are
not numbers of simultaneously recoverable positions. Preferences are
the lists processed by the mechanism.
\end{minipage}
\end{table}

A failure can arise purely from routing: a student eligible for both
reserves is held at $\exc$, which she reaches first, while an academically
eligible but non-disadvantaged student, rejected there, holds a $\dis$
position she is not eligible for; swapping them fills both reserves without
changing who is admitted. Under the deployed order, 25 of the 127
overlapping-reserve program--cycles fail protection-respecting MRU, all in the 2017--2022
transitional period, with a total shortfall of 285 reserve positions
(\autoref{tab:immediate-violations}). Of these, 264 positions are at the 22
failing program--cycles where no rejected student has a claim, so the
shortfall arises mainly from the routing of admitted students; the 63
claimants are all at the remaining three.

Reversing the reserve order increases the number of failing program--cycles to 104 and the shortfall to 962 reserve positions, with failures in every cycle through 2025. The composition also changes: only 371 of these positions are at program--cycles where no rejected student has a claim, and 591 arise where rejected eligible students have claims (an upper
bound on the part of the shortfall that concerns admission), so under the
reverse order the failures affect who is admitted and not only how admitted students are routed. Since the law does not fix the processing order, the asymmetry between the
two orders is a consequence of implementation design: under sequential routing, reserve utilization depends on the reserve order, capacities, and the composition of eligible students. A mechanism that maximizes reserve utilization under either order removes this dependence and lets the order govern only how reserved opportunities are distributed.

The scope of the Chilean problem has narrowed: exposure peaks in 2019 at 35 overlapping-reserve programs in 25 schools, and only three programs at one school remain from 2023, with the end of the transitional academic regime (see \suppref{tab:exposure}). Under the deployed order these show no failure, but under the reverse order all three fail in every cycle through 2025, with shortfalls of 8 to 18 reserve positions per cycle (\autoref{tab:immediate-violations}); the vulnerability belongs to sequential routing with overlapping reserves rather than to the transitional regime, and would reappear at scale if the academic-excellence reserve, still available to authorized schools under the permanent regime (see \suppsectionref{supp:empirical-implementation}), were extended.

\subsection{Sibling-Adjusted Deferred Acceptance}
\label{sec:sibling-adjusted-da}

We propose \textbf{sibling-adjusted deferred acceptance} (SA-DA), which
preserves secured enrollment and the separate PIE procedure while
replacing sequential routing across $\exc$ and $\dis$ with classify-and-rank.
To accommodate sibling priority, it releases just enough $\dis$ positions
to open competition to admit all siblings, so a non-disadvantaged student can
occupy a $\dis$ position only once released.

We now define the choice rule at a program $s$ for an applicant set $I$.
It processes returning students first, then PIE, then the remaining
students by classify-and-rank with a sibling adjustment.

Returning students are taken in a fixed order independent of the
applicant set; each takes the first seat type in her routing order with
remaining capacity, or a personal, nontransferable fallback position if
none remains,\footnote{Fallback positions lie outside the reported capacity; in this
one respect the SA-DA choice rule extends the matchings of
\autoref{sec:model}, and \autoref{prop:sibling-adjusted-outcome} is stated
for this extended domain.} so all are retained.

Remaining PIE positions then go to non-returning students by the
existing PIE procedure, PIE-eligible students in PIE-priority order and
then others in the PIE subschool's order, with unfilled PIE capacity
released to open competition.

The \textbf{residual} students, those neither returning nor selected
through PIE, are selected by classify-and-rank
at the remaining $\exc$, $\dis$, and open capacities, with the open-seat
priority ranking siblings above non-siblings, after releasing to open
competition the smallest number of $\dis$ positions at which
classify-and-rank selects every residual sibling, or all remaining $\dis$
positions if no number does. Released positions are filled in the open
stage, disadvantaged students first and then by open-seat priority; a
disadvantaged holder counts in $\dis$, and other holders are labelled
$\dis^r$. Write $C_s^{\mathrm{SA}}(I)=(\overline C_s^{\mathrm{SA}}(I),\sigma_s^{\mathrm{SA}}(I))$
for the resulting choice rule: $\overline C_s^{\mathrm{SA}}(I)$ consists of
the returning students, the students selected through PIE, and the
selected residual students, and $\sigma_s^{\mathrm{SA}}(I)$ gives their
seat types, with returning students keeping their first-step seat types
and PIE-selected students occupying PIE positions. The formal
construction is in \suppsectionref{supp:sa-da-pi}.

SA-DA processes grades from higher to lower, as in deployment, with
sibling priorities and family-block adjustments determined by its own
higher-grade assignments, and runs student-proposing deferred acceptance
within each grade with the rules $(\overline C_s^{\mathrm{SA}})_s$ and the
processed preference lists. Let $I_s^*$ be the students who propose to
$s$ during the procedure. By path independence (see
\suppsectionref{supp:sa-da-pi}), the students assigned to $s$ are
$\overline C_s^{\mathrm{SA}}(I_s^*)$, with seat types
$\sigma_s^{\mathrm{SA}}(I_s^*)$; seat types are taken from the full proposer
set because the release level and the routing can depend on rejected
students.

Besides protection-respecting MRU (\autoref{def:immediate-violation}), we evaluate the SA-DA outcome against three requirements. An outcome
satisfies \textbf{secured enrollment} if every returning student is
assigned to her current program or to a program she ranks above it. The outcome \textbf{accommodates siblings} at $s$ if, whenever releasing some
number of $\dis$ positions would let classify-and-rank select every sibling
among the residual students who propose to $s$, all of those siblings
are assigned to $s$ and the smallest such number of $\dis$ positions is
released. The outcome satisfies \textbf{priority-respect} at $s$ if it satisfies \autoref{def:priorities} at $s$, with returning students ranked first
for every seat type at their current program, PIE positions ranked by
the order of the PIE procedure, and $\dis^r$ positions treated as open
seats.

\begin{proposition}\label{prop:sibling-adjusted-outcome}\label{prop:sa-da-pi}
Within each grade and conditional on sibling-priority statuses, SA-DA
is strategy-proof, and its outcome satisfies secured enrollment,
sibling accommodation, priority-respect, and
protection-respecting maximum reserve utilization at every
program.\footnote{Like the deployed mechanism, SA-DA is not strategy-proof with
respect to the joint-family application procedure that generates the
sibling priorities and list adjustments taken as given here.}
\end{proposition}

The proof is in \suppsectionref{supp:sa-da-pi}.

\subsection{Reserve Utilization and Distributional Consequences under SA-DA}
\label{sec:sibling-adjusted-empirical}
\label{sec:empirical-order}

We run SA-DA on all nine cycles in two versions and compare both with deployment: academic-first keeps the deployed reserve order for non-returning students, and disadvantaged-first reverses it. Disadvantaged-first SA-DA thus differs from the reverse-order counterfactual of \autoref{sec:res_uti_deployed}, which also reverses the order for returning students.

By \autoref{prop:sibling-adjusted-outcome}, neither version can exhibit a protection-respecting MRU shortfall, and none arises in the data, against 285 positions at 25 program--cycles under deployment (see \suppref{tab:sibling-adjusted-deployment}). At overlapping-reserve programs, eligible holders of reserve positions rise from 15{,}105 under deployment to 15{,}386 under academic-first and 15{,}396 under disadvantaged-first SA-DA, gains of 281 and 291 positions out of 16{,}802; under the reverse-order counterfactual they fall to 14{,}069. Academic-first SA-DA stays very close to deployment: of 3{,}795{,}168
applicant--cycle observations, 89 receive a different program (47
higher-ranked, 42 lower-ranked). Disadvantaged-first changes 1{,}262 (635
higher, 627 lower). Academic-first assigns three fewer students than
deployment and disadvantaged-first five more.

The order nevertheless matters for who is assigned where. Between the two SA-DA
orders, 1{,}185 students receive different programs (1{,}179 different
schools), about 2 percent of the 55{,}857 applicant--cycle observations
listing an overlapping-reserve program (see \suppref{tab:exposure}), and
1{,}774 more keep their program but change seat type (see
\suppref{tab:sibling-adjusted-order-national}). Of the 1{,}185, 591 prefer
academic-first and 594 disadvantaged-first.

To see who benefits, we classify students by disadvantaged status and
academic eligibility, the latter meaning eligibility for $\exc$ at some
overlapping-reserve program on the submitted list (see
\suppsectionref{supp:routing-evidence}); the indicators describe
students, not necessarily the programs at which their assignments change.

\begin{table}[htbp]
\centering
\caption{Reserve order and student eligibility under SA-DA}
\label{tab:sibling-adjusted-order-eligibility}
\small
\setlength{\tabcolsep}{5pt}
\begin{tabular}{lrrrr}
\toprule
Eligibility profile
& Affected
& \begin{tabular}[c]{@{}c@{}}Prefer\\academic first\end{tabular}
& \begin{tabular}[c]{@{}c@{}}Prefer\\disadvantaged first\end{tabular}
& Share (\%) \\
\midrule
Disadvantaged only & 304 & 247 & 57 & 81.3 \\
Academic only & 332 & 60 & 272 & 81.9 \\
Both & 388 & 205 & 183 & -- \\
Neither & 161 & 79 & 82 & -- \\
\midrule
Total & 1{,}185 & 591 & 594 & -- \\
\bottomrule
\end{tabular}

\begin{minipage}{0.90\textwidth}
\footnotesize
\emph{Notes:} Affected students receive different programs
under the two orders. The final column reports the percentage
preferring the order that places the other reserve first.
Preference comparisons use submitted program rankings,
with being unassigned ranked last.
\end{minipage}
\end{table}

The aggregate balance masks a sharp split: 81.3 percent of affected
students who are disadvantaged but not academically eligible prefer
academic-first, and 81.9 percent of those academically eligible but not
disadvantaged prefer disadvantaged-first, so students eligible for one reserve category tend to benefit when the other category comes first.
Students eligible for both reserves or for neither are close
to evenly split, 205 to 183 and 79 to 82; those eligible for
neither can be affected only through rejection chains.

This matches \autoref{prop:access}: with $\exc$ first, dual-eligible
students may be routed through $\exc$, leaving $\dis$ positions for
disadvantaged-only students, and conversely. \autoref{prop:access}
compares reserve choices at a fixed pool, however, whereas the national
comparison also incorporates rejection chains and Chile's institutional
features.


\section{Conclusion}\label{sec:conclusion}

By selecting within common-priority classes before fixing categories, classify-and-rank achieves maximum reserve utilization, priority-respect, and strategy-proofness jointly even when overlapping reserves are governed by different priorities. These requirements do not determine how access is distributed; the common-priority classification and reserve order resolve that multiplicity, with a behavioral foundation for two reserve categories and a characterization for any number, and make the distributional choices in reserve design explicit.

The Chilean application shows that sequential reserve routing can generate avoidable reserve-utilization failures whose frequency depends substantially on the reserve order. We develop an institutionally adapted mechanism that is strategy-proof within each grade, taking sibling priorities as given, and achieves maximum reserve utilization subject to Chile's institutional protections, regardless of the order. Our counterfactual analysis shows that these failures can be addressed with relatively few changes to existing school
assignments. At the same time, reserve order has clear distributional implications: students eligible for only one reserve category tend to benefit when the other category comes first. Although the overlap between Chile's two reserves has narrowed in recent cycles, the design problem is general: Indian state admissions administer many overlapping categories under distinct criteria, where the common-priority classification applies directly.

\bibliographystyle{aer}
\bibliography{matching}

\appendix
\renewcommand{\sectionautorefname}{Appendix}
\renewcommand{\subsectionautorefname}{Appendix}

\section{Single-Priority Approach and Its Limits}
\label{sec:single_priority}

We recover the common-priority benchmark and its characterization, then show
why Gale dominance provides no uniform selection principle when priorities
differ.

Fix a school $s\in\mathcal S$ and assume $\mathcal T\neq\varnothing$. We say
that $s$ faces a \textbf{single-priority environment} if all seat types are
governed by the same criterion. Write the resulting common strict priority
over $\mathcal I$ as $\succ_s$, so that $\succ_s^t=\succ_s$ for every
$t\in\mathcal T\cup\{o\}$.

With $\mathcal B_s=\{\mathcal T\}$, classify-and-rank selects reserve
students by the common priority subject to maximum utilization and fills
the remaining capacity by the same priority; the reserve order affects only
routing.

For an applicant set $I\subseteq\mathcal I$, a selected student set
$J\subseteq I$ with $|J|\leq q_s$ satisfies \textbf{no justified envy}
\citep{sonmez/yenmez:22} if, for every $i\in I\setminus J$ and every
$j\in J$ with $i\succ_s j$, $n_s((J\setminus\{j\})\cup\{i\})<n_s(J)$.
A restricted choice rule satisfies this condition if its selected set does so
for every applicant set $I$.

We say that $J$ \textbf{supports maximum reserve utilization} if
$n_s(J)=n_s(I)$.

\begin{corollary}\label{cor:single-priority}
In a single-priority environment with
$\mathcal B_s=\{\mathcal T\}$:
\begin{enumerate}[label=(\roman*)]
\item
$\overline{\CR}_s$ coincides with the meritorious horizontal choice rule of
\citet{sonmez/yenmez:22};

\item
joint priority-respect is equivalent to priority-respect together with
the no-justified-envy condition; and

\item
$\overline{\CR}_s$ is the unique restricted choice rule that is
non-wasteful, satisfies no justified envy, and supports maximum reserve
utilization.
\end{enumerate}
\end{corollary}

Priority-respect appears separately in part~(ii) because no justified envy
constrains only the selected set: a set satisfying it can still be routed so
that a lower-priority student holds an open seat that an excluded student
could claim. When the routing fills $n_s(J)$ reserved seats, no justified
envy rules this out, because the swap would leave reserve utilization
unchanged. Classify-and-rank routes this way; part~(iii) concerns selected
sets only.

The proof is in \suppsectionref{supp:single-priority}; part~(iii)
recovers \citet[Theorem~2]{sonmez/yenmez:22}. We next show that Gale
dominance does not extend uniformly to reserve-specific priorities.
Fix a school $s\in\mathcal S$ and an applicant set
$I\subseteq\mathcal I$, and consider matchings $(I_s,\sigma_s)$ for $s$ with
$I_s\subseteq I$. Such a matching is \textbf{priority-respecting from $I$} if,
for every $i\in I_s$ with $\sigma_s(i)=t$ and every $j\in I\setminus I_s$,
either $j\notin\mathcal I^t$ or $i\succ_s^t j$; it satisfies
\textbf{maximum reserve utilization from $I$} if there exist no
reserve-eligible student $i\in I\setminus I_s$ and matching
$(I_s',\sigma_s')$ for $s$ such that $i\in I_s' \subseteq I_s\cup\{i\}$
and
$\rho_s(I_s',\sigma_s')>\rho_s(I_s,\sigma_s)$.
These are the school-level counterparts of \autoref{def:priorities} and
\autoref{def:maximum_reserve_utilization}. Let $\PR_s(I)$ and $\RE_s(I)$
denote the corresponding collections of matchings.

For a strict order $\succ$ and two subsets $I',I''\subseteq\mathcal I$, list
their members from highest to lowest according to $\succ$ as
$I'=\{i'_1,\ldots,i'_{|I'|}\}$ and $I''=\{i''_1,\ldots,i''_{|I''|}\}$. We say
that $I'$ \textbf{Gale-dominates} $I''$ with respect to $\succ$ if
$|I'|\geq|I''|$ and $i'_k\succeq i''_k$ for every $k\leq|I''|$, where $\succeq$
denotes the weak extension of $\succ$. When the reference order is the
open-seat priority $\succ_s^o$, we simply say that $I'$ Gale-dominates
$I''$.

For a set $\mathcal A$ of matchings for $s$ from $I$, a matching
$(I_s,\sigma_s)\in\mathcal A$ is \textbf{Gale-dominant} in $\mathcal A$ with
respect to $\succ$ if $I_s$ Gale-dominates, with respect to $\succ$, the
student set of every other matching in $\mathcal A$.

\begin{proposition}\label{prop_nonexistence_gale}
Fix a finite student set $\mathcal I$ with $|\mathcal I|\geq 6$. There
exists an environment with student set $\mathcal I$, two reserve
categories, and school-specific reserve eligibility such that, for every
strict ranking $\succ$ over $\mathcal I$, there exist a school $s$ and an
applicant set $I\subseteq\mathcal I$ for which no matching in
$\RE_s(I)\cap\PR_s(I)$ is Gale-dominant in that set with respect to
$\succ$.
\end{proposition}

The environment in \appref{app:gale_dominance} has six schools sharing a
six-student pool and differing only in reserve eligibility and priorities;
no matching in the relevant sets uses an open seat, so the conclusion does
not depend on open-seat priorities. Gale dominance need not fail in every
environment.\footnote{Selecting a Gale-dominant matching whenever one
exists need not guarantee path independence; see \suppref{supp:gale-path}.}


\subsection{Proof of \autoref{prop_nonexistence_gale}}
\label{app:gale_dominance}

Choose six students $i_1,\ldots,i_6\in\mathcal I$, let
$I\coloneqq \{i_1,\ldots,i_6\}$ and $W\coloneqq \{i_3,i_4,i_5,i_6\}$, and take two reserve
categories, $t_1$ and $t_2$. For each of the three partitions of $W$ into pairs, $(P,Q)$
with $P\in\{\{i_3,i_6\},\{i_3,i_4\},\{i_3,i_5\}\}$ and $Q=W\setminus P$,
the environment has two schools, each with capacity $4$, two seats in each
category, an arbitrary open-seat priority, $I^1=\{i_1,i_2\}\cup P$,
$I^2=\{i_1,i_2\}\cup Q$, and $i_2\succ^1 i_1$ and $i_2\succ^2 i_1$ ranked
above $P$ under $\succ^1$ and above $Q$ under $\succ^2$. The two schools
of a partition differ only in the order of $P$ under $\succ^1$ and of $Q$
under $\succ^2$, which the second reverses. Priorities are completed
arbitrarily, and students outside $I$ do not enter the pool.

\smallskip
\noindent\emph{Step 1: Selected sets at a school.}
Fix one of the six schools and let $p^*$ and $q^*$ be the $\succ^1$-top
member of $P$ and the $\succ^2$-top member of $Q$. Since $n(I)=4$ and every
matching from $I$ excludes at least two reserve-eligible students,
\autoref{lem:mru-demand} implies that every matching in
$\RE(I)\cap\PR(I)$ fills all four reserve seats and uses no open seat. Each
admits $i_2$ and $i_1$, who rank first and second at both reserves:
excluding either leaves a lower-priority occupant in a category she can
claim. If $i_1,i_2$ both occupy $t_1$, then $t_2$ admits $Q$; if both occupy
$t_2$, then $t_1$ admits $P$; if they are split, priority-respect
places $p^*$ and $q^*$ in the remaining seats. All three routings respect
priorities and maximize utilization, so the selected sets are exactly
$X=\{i_1,i_2\}\cup P$, $Y=\{i_1,i_2,p^*,q^*\}$, and $Z=\{i_1,i_2\}\cup Q$.

\smallskip
\noindent\emph{Step 2: Reduction to the pairs.}
Fix a strict ranking $\succ$ over $\mathcal I$. For sets $I',I''$ of equal
size, $I'$ Gale-dominates $I''$ with respect to $\succ$ if and only if, for
every $m$, $I'$ contains at least as many of the $m$ highest-ranked students
as $I''$: if $i'_k\succeq i''_k$ for all $k$, then $i''_k$ among the top $m$
implies $i'_k$ is too; if $i''_k\succ i'_k$ for some $k$, then with $m$ the
rank of $i''_k$, $I''$ contains at least $k$ of the top $m$ and $I'$ at most
$k-1$. Since $X$, $Y$, and $Z$ share $i_1$ and $i_2$, dominance among them
is decided by their two remaining members.

\smallskip
\noindent\emph{Step 3: Choosing the school.}
Let $w_1\succ w_2\succ w_3\succ w_4$ be the members of $W$ ordered by
$\succ$. Exactly one partition puts $w_1$ and $w_4$ in the same pair; take
it, and take the school at which that pair's reserve priority ranks $w_4$
above $w_1$, so that $p^*=w_4$ if the pair is $P$ and $q^*=w_4$ if it is
$Q$. Since $w_1\succ w_2$ but $w_3\succ w_4$, the pairs $\{w_1,w_4\}$ and
$\{w_2,w_3\}$ are Gale-incomparable, so by Step~2 neither $X$ nor $Z$ is
Gale-dominant; and $Y$, whose two remaining members are $w_4$ and one of
$w_2,w_3$, does not dominate the set containing $w_1$. Hence no matching in
$\RE_s(I)\cap\PR_s(I)$ at this school is Gale-dominant with respect to
$\succ$.\qed

\section{Proof of \autoref{thm:existence}}
\label{supp:existence}

We first establish school-level and choice-theoretic properties of
classify-and-rank, some used again later, then apply cumulative-offer
results and establish joint priority-respect.

\subsection{Properties of Classify-and-Rank}
\label{app:cnr-properties}

Fix a school $s\in\mathcal S$ and an applicant set $I\subseteq\mathcal I$.
Let $\Hset_s(I)$ denote the set of matchings $(I_s,\sigma_s)$ for $s$ with
$I_s\subseteq I$. Recall from \appref{sec:single_priority} the notions
of being \emph{priority-respecting from $I$} and \emph{maximizing reserve
utilization from $I$}. A matching $(I_s,\sigma_s)$ is \textbf{non-wasteful
from $I$} if $|I_s|=\min\{q_s,|I|\}$. Let $\PR_s(I)$, $\RE_s(I)$, and
$\NW_s(I)$ denote the corresponding subsets of $\Hset_s(I)$.

For the first two stages we use the following notation. Let
$F_s(I)\coloneqq \{(k,i)\in\{1,\ldots,K_s\}\times I: r_s^t>0\text{ for some }t\in B_s^k\cap\tau(i)\}$
be the feasible class--student pairs and
$\mathcal Z_s(I)\coloneqq \{\pi_s(X):X\in\mathcal X_s(I)\}$ the realizable class
commitments. The orders of Stages~1 and~2 are $(k,i)\vdash_s(k',j)$ if and
only if $k<k'$, or $k=k'$ and $i\succ_s^{B_s^k}j$; and $(t,i)\vdash_s^R(t',j)$
if and only if $t\triangleright_s t'$, or $t=t'$ and $i\succ_s^t j$.

The first result records the greedy structure of the first two stages of
classify-and-rank.

\begin{lemma}[Lexicographic maximality]
\label{lem:lexicographic_dominance}
For every school $s$ and applicant set $I$:
\begin{enumerate}[label=(\roman*)]
\item $Z_s(I)$ is the unique lexicographically maximal element, under
$\vdash_s$, of $\{Z\in\mathcal Z_s(I):|Z|=n_s(I)\}$. In particular,
$Z_s(I)$ is realizable and $|Z_s(I)|=n_s(I)$.

\item $X_s(I)$ is the unique lexicographically maximal element, under
$\vdash_s^R$, of
$\mathcal X_s(I;Z_s(I))\coloneqq \{X\in\mathcal X_s(I):\pi_s(X)=Z_s(I)\}$.
\end{enumerate}
\end{lemma}

\begin{proof}
Consider Stage~1. After each pair is processed, the retained set is
contained in some realizable class commitment of size $n_s(I)$, by the
definition of the extendability test. Let $Z$ be such an extension of the
final retained set. If the final set were a proper subset of $Z$, take a
pair in $Z$ that was not retained. When that pair was considered, $Z$
itself would have provided a size-$n_s(I)$ extension of the previously
retained pairs together with that pair, so it would have been retained, a
contradiction. Thus the final retained set $Z_s(I)$ is realizable and has
size $n_s(I)$.

Now let $Z'\neq Z_s(I)$ be any other realizable class commitment of size
$n_s(I)$, and let $f$ be the first pair under $\vdash_s$ at which they
differ. If $f\in Z'\setminus Z_s(I)$, then $Z'$ would have certified the
extendability of $f$ when it was considered, so $f$ would have been
retained. Hence $f\in Z_s(I)\setminus Z'$, and $Z_s(I)$ is
lexicographically superior to $Z'$.

The same argument applied to the Stage~2 extendability procedure,
conditional on $Z_s(I)$ and using the order $\vdash_s^R$, proves part~(ii).
\end{proof}

\begin{proposition}[School-level properties]
\label{prop:greedy}
For every school $s$, applicant set $I$, and classify-and-rank specification,
$\CR_s(I)\in\PR_s(I)\cap\RE_s(I)\cap\NW_s(I)$.
\end{proposition}

\begin{proof}
Write $\CR_s(I)=(I_s,\sigma_s)$. By construction, the reserve stages use
$n_s(I)$ reserved seats, and Stage~3 fills any remaining capacity whenever
enough students are available. Hence
$\CR_s(I)\in\RE_s(I)\cap\NW_s(I)$.

For priority-respect, suppose toward a contradiction that there exist
$i\in I_s$, $j\in I\setminus I_s$, and $t\in\mathcal T\cup\{o\}$ with
$\sigma_s(i)=t$, $j\in\mathcal I^t$, and $j\succ_s^t i$.

Suppose first that $t\in\mathcal T$, and let $B_s^k$ be the class
containing $t$. Replacing $(t,i)$ by $(t,j)$ in $X_s(I)$ gives a feasible
reserve assignment of the same cardinality, and $(k,j)\vdash_s(k,i)$.
Because $(k,j)$ precedes $(k,i)$, the replacement leaves intact every pair
retained before $(k,j)$ was considered. It therefore provides a
size-$n_s(I)$ extension of those pairs together with $(k,j)$, so Stage~1
would have retained $(k,j)$, contradicting $j\notin I_s$.

If $t=o$, then $i$ is selected in Stage~3 while $j$ remains available for
an open seat. Since Stage~3 selects the highest-$\succ_s^o$ remaining
students and $j\succ_s^o i$, this is again impossible.
\end{proof}

\subsection{Restricted Choice and Cumulative Offers}

Let $\overline{\CR}_s(I)$ denote the set of students selected by
$\CR_s(I)$, ignoring their seat types. We establish the properties of
$\overline{\CR}_s$ needed to apply the cumulative-offer results of
\citet{hatfi04}.

A restricted choice rule $C_s$ is \textbf{substitutable} if
$C_s(I')\cap I\subseteq C_s(I)$ whenever $I\subseteq I'$, and is
\textbf{size monotone} if $|C_s(I)|\leq|C_s(I')|$ whenever $I\subseteq I'$.
It is \textbf{path independent} \citep{plott} if
$C_s(I\cup I')=C_s(C_s(I)\cup I')$ for all $I,I'\subseteq\mathcal I$.

\begin{lemma}
\label{substitutability}
For every school $s$, $\overline{\CR}_s$ is substitutable.
\end{lemma}

\begin{proof}
Fix $s$. We construct a single assignment valuation, with jobs,
productivities, and wages fixed independently of the applicant set, whose
demand from every applicant set $I$ is $\overline{\CR}_s(I)$.

List $F_s(\mathcal I)=\{f_1,\ldots,f_m\}$ with
$f_1\vdash_s\cdots\vdash_s f_m$, and
$E_s(\mathcal I)=\{e_1,\ldots,e_n\}$ with
$e_1\vdash_s^R\cdots\vdash_s^R e_n$. Choose positive integers
$(u_i)_{i\in\mathcal I}$ representing $\succ_s^o$, so that
$i\succ_s^o j$ if and only if $u_i>u_j$. Let $V\coloneqq \sum_{i\in\mathcal I}u_i$ and $\omega\coloneqq 2^{m+n}(V+1)$.

Construct $q_s$ jobs: for each $t\in\mathcal T$, there are $r_s^t$
reserve jobs $(t,1),\ldots,(t,r_s^t)$, together with
$q_s-\sum_t r_s^t$ ordinary open jobs. Every student can fill every job.
Whenever $t\in\tau(i)$ and $r_s^t>0$, let $k(t,i)$ and $\ell(t,i)$ be the
positions of the corresponding class--student and reserve--student pairs
in the two orders above. Define

$$
a_{ij}\coloneqq u_i+
\begin{cases}
\omega^3+\omega^2 2^{m-k(t,i)}+\omega 2^{n-\ell(t,i)},
&\text{if }j=(t,p)\text{ and }t\in\tau(i),\\
0,
&\text{otherwise.}
\end{cases}
$$

Give job $j$ productivity $1+a_{ij}$ when filled by student $i$, and offer
each student the single fixed-wage contract with wage one. A demanded set
never contains an unassigned student, since dropping such a student saves
the wage without reducing productivity. Thus demand is determined by a
maximum-weight assignment under the weights $(a_{ij})$.

For an assignment $\nu$, let $A(\nu)$ denote its assigned students and
let $X(\nu)$ be the eligible reserve--student pairs it uses. Define
$\lambda(Z)\coloneqq \sum_{{k:f_k\in Z}}2^{m-k}$ and
$\eta(X)\coloneqq \sum_{{\ell:e_\ell\in X}}2^{n-\ell}$. The total net weight is

$$
W(\nu)
=
\omega^3|X(\nu)| + \omega^2\lambda\bigl(\pi_s(X(\nu))\bigr) + \omega\,\eta\bigl(X(\nu)\bigr)
+
\sum_{i\in A(\nu)}u_i.
$$

Since $\omega>V$, $\omega^2>\omega(2^n-1)+V$, and $\omega^3>\omega^2(2^m-1)+\omega(2^n-1)+V$ (each follows from $\omega\ge 2^{m+n}$ and $\omega>V$), maximizing $W$ first maximizes $|X(\nu)|$, then the class commitment
lexicographically under $\vdash_s$, then the reserve assignment
lexicographically under $\vdash_s^R$. By
\autoref{lem:lexicographic_dominance}, the resulting reserve assignment is
therefore $X_s(I)$.

Let $S$ be the set of students assigned by $X_s(I)$. If a reserve category is underfilled in $X_s(I)$, no student in
$I\setminus S$ is eligible for that category; otherwise that student could
be added to the underfilled category, contradicting maximality of reserve
utilization. Hence, conditional on $X_s(I)$, every remaining job yields
only the base weight $u_i$ to every remaining student. Since the $u_i$ are
positive and represent $\succ_s^o$, the remaining jobs are filled, up to
capacity, by precisely the highest-$\succ_s^o$ students in $I\setminus S$.
Thus every maximum-weight assignment selects exactly
$\overline{\CR}_s(I)$.

The construction is fixed independently of $I$ and defines an assignment
valuation. Assignment valuations are substitutes valuations
\citep[Theorem~14]{hatfi04}. Since the demanded student set is uniquely
$\overline{\CR}_s(I)$ for every $I$, the restricted choice rule
$\overline{\CR}_s$ is substitutable.
\end{proof}

\begin{lemma}
\label{size_monotonicity}\label{path_independence}
For every school $s$, $\overline{\CR}_s$ is size monotone and path independent.
\end{lemma}

\begin{proof}
By \autoref{prop:greedy},
$|\overline{\CR}_s(I)|=\min\{q_s,|I|\}$, giving size monotonicity.
If $i\notin\overline{\CR}_s(I)$, substitutability gives
$\overline{\CR}_s(I)\subseteq\overline{\CR}_s(I\setminus\{i\})$,
and size monotonicity makes their cardinalities equal. Thus the choices
coincide, establishing consistency. Consistency and substitutability imply
path independence \citep{AizMal:1981AutCont}.
\end{proof}

We next relate DA based on the full choice rules $\CR$ to the
cumulative-offer mechanism based on the restricted choice rules
$\overline{\CR}$.

\begin{corollary}
\label{cor:da_co_seats}
Let $(\mu,\sigma)$ be the outcome of DA based on $\CR$, and let $I_s^*$
be the set of students who propose to school $s$ during the procedure.
Then DA based on $\CR$ and the cumulative-offer mechanism based on
$\overline{\CR}$ induce the same school matching, and
$\CR_s(I_s^*)=(\mu^{-1}(s),\sigma_s)$ for every $s$.
\end{corollary}

\begin{proof}
Let $I_s(k)$ and $P_s(k)$ be the cumulative and new proposals received by
$s$ through and at step $k$. Path independence gives
$\overline{\CR}_s(I_s(k-1)\cup P_s(k))=\overline{\CR}_s(\overline{\CR}_s(I_s(k-1))\cup P_s(k))$,
so the cumulative-offer and DA updates coincide at every step and
$\overline{\CR}_s(I_s^*)=\mu^{-1}(s)$. For seat types, let
$J_s\subseteq I_s^*$ be the pool at the final DA update, so that
$\CR_s(J_s)=(\mu^{-1}(s),\sigma_s)$. Since $X_s(I_s^*)$ uses only students
in $\mu^{-1}(s)\subseteq J_s$, $n_s(J_s)=n_s(I_s^*)$, and every commitment
realizable from $J_s$ is realizable from $I_s^*$; hence
\autoref{lem:lexicographic_dominance} gives $Z_s(J_s)=Z_s(I_s^*)$ and, by
the same argument at Stage~2, $X_s(J_s)=X_s(I_s^*)$. The common selected
set then determines identical open-seat assignments.
\end{proof}

\subsection{Concluding the Proof}

\begin{lemma}[Greedy basis plus top-up]
\label{lem:greedy_topup}
Let $\mathcal N$ be a matroid on a finite ground set $U$ with rank $r$, and
let $\succ$ be a strict order on $U$. Let $H$ be the greedy basis of
$\mathcal N$ under $\succ$, and let $O$ consist of the $m$
highest-$\succ$ elements of $U\setminus H$. If
$i\in U\setminus(H\cup O)$, $j\in H\cup O$, and $i\succ j$, then
$\operatorname{rank}_{\mathcal N}((H\cup O\setminus\{j\})\cup\{i\})<r$.

\end{lemma}

\begin{proof}
The matroid greedy theorem implies that $H\cup O$ Gale-dominates every
$Y\subseteq U$ with $|Y|=|H|+|O|$ and
$\operatorname{rank}_{\mathcal N}(Y)=r$. To see this, apply the
greedy-basis theorem to the dual matroid under the reverse order:
$U\setminus H$ is its greedy basis, and $U\setminus(H\cup O)$ is an
initial segment of that basis. Taking complements gives the claim.

If $Y\coloneqq (H\cup O\setminus\{j\})\cup\{i\}$ had rank $r$, then $H\cup O$
would Gale-dominate $Y$. But $Y$ replaces $j$ by the strictly higher-ranked
element $i$, so $Y$ strictly Gale-dominates $H\cup O$, a contradiction.
\end{proof}

Individual rationality follows from DA. By \autoref{substitutability}
and \autoref{size_monotonicity}, the restricted rules are substitutable
and size monotone. The cumulative-offer mechanism is therefore
strategy-proof \citep{hatfi04}, and \autoref{cor:da_co_seats} transfers
this property to DA based on $\CR$.

Fix its outcome $M=(\mu,\sigma)$ and a school $s$. By
\autoref{cor:da_co_seats}, its terminal choice is $\CR_s(I_s^*)$, where
$I_s^*$ contains all students who proposed to $s$, including every student
who prefers $s$ to her assignment. Non-wastefulness and priority-respect
therefore follow from \autoref{prop:greedy}. Any deviation violating
maximum reserve utilization would use students from
$\mu^{-1}(s)\cup\{i\}\subseteq I_s^*$ and more than $n_s(I_s^*)$
reserved seats, which is impossible.

For joint priority-respect, fix $B\in\mathcal B_s$, write
$H\coloneqq A_s^B(M)$, let $Q$ be the students assigned by $X_s(I_s^*)$
other than those in $H$, and write $U\coloneqq I_s^*\setminus Q$. We hold the reserve assignments of $Q$ fixed.

Suppose first that open seats do not share $B$'s criterion, so the relevant
group is $G=B$ and $n_s^B(H)=|H|$. A violation would give an excluded
student $i\in I_s^*$ and $j\in H$ with $i\succ_s^B j$ and $n_s^B((H\setminus\{j\})\cup\{i\})=|H|$. The replacement set can be assigned through $B$. Together with $Q$'s
unchanged assignments, it induces a realizable commitment of size
$n_s(I_s^*)$ that improves $Z_s(I_s^*)$ within $B$ and agrees elsewhere,
contradicting \autoref{lem:lexicographic_dominance}.

Now suppose open seats share $B$'s criterion. Write
$\succ_s\coloneqq \succ_s^B=\succ_s^o$, $G\coloneqq B\cup\{o\}$, and
$O\coloneqq A_s^o(M)$, so $A_s^G(M)=H\cup O$.
Let $\mathcal N_B$ be the transversal matroid on $U$ presented by
$r_s^t$ copies of $U\cap\mathcal I^t$ for each $t\in B$.
Its independent sets are exactly those assignable through $B$, so
$\operatorname{rank}_{\mathcal N_B}(Y)=n_s^G(Y)$ for $Y\subseteq U$.

This matroid has rank $|H|$: $H$ is independent, and any larger independent
set, together with $Q$'s fixed assignments, would increase reserve utilization.
Moreover, $H$ is its greedy basis under $\succ_s$. A lexicographically
superior basis $H'\subseteq U$ could replace $H$ while retaining $Q$,
giving a realizable size-$n_s(I_s^*)$ class commitment superior to
$Z_s(I_s^*)$, again contradicting \autoref{lem:lexicographic_dominance}.
Stage~3 selects $O$ as the highest-$\succ_s$ students in $U\setminus H$
up to the remaining capacity.

A joint-priority violation would now give
$i\in U\setminus(H\cup O)$ and $j\in H\cup O$ with $i\succ_s j$ and $n_s^G((H\cup O\setminus\{j\})\cup\{i\})=n_s^G(H\cup O)=|H|$. This contradicts \autoref{lem:greedy_topup} applied to $\mathcal N_B$.

Thus the additional joint-priority condition holds for every
$B\in\mathcal B_s$. Since $s$ was arbitrary, the DA mechanism based on
$\CR$ satisfies joint priority-respect. Together with the preceding
arguments, this proves \autoref{thm:existence}.

\section{Proof of \autoref{thm:characterization}}
\label{supp:characterization}

For a school matching $\mu$ and school $s$, let
$\mathcal D_s(\mu)\coloneqq \{i\in\mathcal I:s\mathrel{R_i}\mu(i)\}$.

We first record an equivalent formulation of maximum reserve utilization.

\begin{lemma}\label{lem:mru-demand}
Fix a school $s$, an applicant set $I\subseteq\mathcal I$, and a matching
$M=(I_s,\sigma_s)$ for $s$ with $I_s\subseteq I$. Then $M$ satisfies
maximum reserve utilization from $I$ if and only if there exist no
reserve-eligible student $i\in I\setminus I_s$ and matching
$(I_s',\sigma_s')$ for $s$ with $i\in I_s'\subseteq I$ and
$\rho_s(I_s',\sigma_s')>\rho_s(I_s,\sigma_s)$.
\end{lemma}

\begin{proof}
Only the ``only if'' direction requires proof. Write
$r\coloneqq \rho_s(I_s,\sigma_s)$ and suppose the stated condition fails. Thus
$n_s(I)>r$ and there is a reserve-eligible student outside $I_s$.

If $n_s(I_s)>r$, take any such student $i$. Since $i$ is a nonloop of
the transversal matroid induced by the reserve categories, it belongs to
a basis of $I_s\cup\{i\}$, whose size is at least $n_s(I_s)>r$.
If instead $n_s(I_s)=r$, matroid augmentation and $n_s(I)>r$ give some
$i\in I\setminus I_s$ with $n_s(I_s\cup\{i\})>r$. In either case there
is a matching contained in $I_s\cup\{i\}$, containing $i$, and using
more than $r$ reserved seats, contradicting maximum reserve utilization
from $I$.
\end{proof}

Since preferences are strict,
$\mathcal D_s(\mu)\setminus\mu^{-1}(s)=\{i:s\mathrel{P_i}\mu(i)\}$. Hence
\autoref{lem:mru-demand} implies that matching-level maximum reserve
utilization at $s$ is equivalent to maximum reserve utilization from
$\mathcal D_s(\mu)$.

For a matching $M=(I_s,\sigma_s)$ for $s$ drawn from $I$, define
\textbf{class precedence from $I$} by conditions~(i)--(v) of
\autoref{def:class_precedence}, taking $M'$ to be any matching for $s$
drawn from $I$ and replacing $\mu^{-1}(s)$ by $I_s$.
In particular, condition~(v) requires $I_s'\neq I_s$.
With $I=\mathcal D_s(\mu)$, this is exactly class precedence at $s$ for
$(\mu,\sigma)$.

For $r=0,\ldots,n_s(I)$, let
$\mathcal Z_s^r(I)\coloneqq \{Z\in\mathcal Z_s(I):|Z|=r\}$.
These sets are nonempty: any $r$-edge subassignment of a maximum reserve
assignment induces an element of $\mathcal Z_s^r(I)$. Let $Z_s^r(I)$
denote the unique lexicographically maximal element of
$\mathcal Z_s^r(I)$ under $\vdash_s$, and write
$S_s^r(I)\coloneqq \{i:(k,i)\in Z_s^r(I)\text{ for some }k\}$.
In particular, $Z_s^{n_s(I)}(I)=Z_s(I)$.

\begin{lemma}\label{lem:nested-class-access}
For every school $s$, applicant set $I$, and $r<n_s(I)$,
$S_s^r(I)\subseteq S_s^{r+1}(I)$.
\end{lemma}

\begin{proof}
Expand every reserve category into its distinct slots, and choose slot
matchings $X^r$ and $X^{r+1}$ realizing $Z_s^r(I)$ and
$Z_s^{r+1}(I)$, respectively. Suppose some
$i\in S_s^r(I)\setminus S_s^{r+1}(I)$, and let $P$ be the component of
$X^r\triangle X^{r+1}$ containing $i$.

If $P$ contains equally many edges from the two matchings, switching along
$P$ preserves both cardinalities. At the first class--student pair under
$\vdash_s$ whose membership changes, one of the switched matchings is
lexicographically superior to the corresponding $Z_s^r(I)$ or
$Z_s^{r+1}(I)$, a contradiction.

Otherwise $P$ contains one more edge from $X^r$. Since
$|X^{r+1}|=|X^r|+1$, another component $Q$ contains one more edge from
$X^{r+1}$. Switching along $P\cup Q$ again preserves both cardinalities,
and the first changed class--student pair yields the same contradiction.
Thus $S_s^r(I)\subseteq S_s^{r+1}(I)$.
\end{proof}

\begin{lemma}\label{lem:access_choice}
For every school $s$ and applicant set $I$, $\CR_s(I)$ respects class
precedence from $I$. Moreover, if
$M=(I_s,\sigma_s)\in\NW_s(I)\cap\PR_s(I)\cap\RE_s(I)$ respects class
precedence from $I$, then $I_s=\overline{\CR}_s(I)$.
\end{lemma}

\begin{proof}
First consider $\CR_s(I)$. If $B_s^k$, $i$, and $M'$ witnessed a
violation of class precedence, then, since $\CR_s(I)$ uses $n_s(I)$
reserved seats, $M'$ would also use $n_s(I)$ reserved seats.
Conditions~(ii)--(iv) would make its class commitment lexicographically
superior to $Z_s(I)$, contradicting
\autoref{lem:lexicographic_dominance}. Thus $\CR_s(I)$ respects class
precedence.

Now let $M=(I_s,\sigma_s)$ satisfy the hypotheses. Write
$C\coloneqq \overline{\CR}_s(I)$, $c\coloneqq \min\{q_s,|I|\}$, let $Z$ be the class
commitment of the reserve assignment of $M$, and set $m\coloneqq |Z|$ and
$n\coloneqq n_s(I)$. Non-wastefulness gives $|I_s|=|C|=c$. Suppose
$I_s\neq C$.

\medskip
\noindent\emph{Case 1: $m=n$.}

If $Z\neq Z_s(I)$, then $Z_s(I)$ is lexicographically superior to $Z$.
At their first differing pair, $\CR_s(I)$ satisfies conditions
(i)--(iv) of class precedence; condition~(v) follows from $I_s\neq C$.
This contradicts class precedence. Hence $Z=Z_s(I)$.

Thus $M$ and $\CR_s(I)$ select the same students through reserves.
Choose $x\in C\setminus I_s$ and $y\in I_s\setminus C$. Neither is
reserve-selected. Since Stage~3 selects $x$ but not $y$, we have
$x\succ_s^o y$, contradicting priority-respect. Hence $I_s=C$.

\medskip
\noindent\emph{Case 2: $m<n$.}

We first claim that every reserve-eligible student in $I$ belongs to
$I_s$; the claim uses the reformulation in \autoref{lem:mru-demand}, which
lets the alternative matching draw on all of $I$. Otherwise take a reserve-eligible $i\in I\setminus I_s$ and a maximum reserve
assignment. Every category for which $i$ is eligible is full in it, since
otherwise adding $i$ would enlarge it, so replacing one occupant of such a
category by $i$ gives a maximum reserve assignment containing $i$. A matching
drawn from $I$ therefore contains $i$ and uses $n>m$ reserved seats,
contradicting \autoref{lem:mru-demand}. In
particular $S_s^n(I)\subseteq I_s$.

If $c=n$, then $C=S_s^n(I)$, so this inclusion and $|I_s|=|C|$ imply $I_s=C$.
Hence $m<n<c$.

Choose $x\in C\setminus I_s$. By the preceding claim, $x$ is not
reserve-eligible and therefore enters $C$ at Stage~3. For any
$y\in I_s\setminus C$, $y\notin S_s^n(I)\subseteq C$, so $y$ is also
available at Stage~3 but is not selected. Hence $x\succ_s^o y$.
Since every student is eligible for open seats, priority-respect implies
that every such $y$ must be reserve-assigned under $M$. Thus $I_s\setminus C\subseteq\{i:(k,i)\in Z\text{ for some }k\}$.

We claim that $Z=Z_s^m(I)$. Otherwise $Z_s^m(I)$ is lexicographically
superior to $Z$. Let $\widehat X$ realize $Z_s^m(I)$. Every student in
$\widehat X$ is reserve-eligible and hence belongs to $I_s$. Since
$m<c$, choose $z\in I_s$ not used by $\widehat X$. Replacing the reserve assignment of $M$ by $\widehat X$, dropping $z$, and
admitting $x$ through an open seat gives an alternative with the same reserve
utilization and a different student set; it is feasible because it admits
$c\le q_s$ students and open seats are bounded only by capacity. Let $(k,i)$
be the first pair under $\vdash_s$ at which $Z_s^m(I)$ and $Z$ differ, so
that $(k,i)\in Z_s^m(I)\setminus Z$. Since $\vdash_s$ orders pairs first by
class, $Z$ and $Z_s^m(I)$ agree on classes before $k$, giving
condition~(ii); pairs of class $k$ ranked above $i$ precede $(k,i)$,
giving~(iii); and $(k,i)\in Z_s^m(I)$ gives~(iv). Conditions~(i) and~(v)
hold by construction, so this alternative witnesses a violation of class
precedence.
Therefore $Z=Z_s^m(I)$.

Applying \autoref{lem:nested-class-access} $n-m$ times gives
$S_s^m(I)\subseteq S_s^n(I)\subseteq C$. Together with $I_s\setminus C\subseteq\{i:(k,i)\in Z\text{ for some }k\}$ and $Z=Z_s^m(I)$, this
implies $I_s\setminus C\subseteq C$, impossible because
$I_s\setminus C$ is nonempty and disjoint from $C$. Hence $I_s=C$.
\end{proof}

For each $s$, let $\psi_s^C(I)$ consist of the sets $I_s\subseteq I$
admitting a routing in $\NW_s(I)\cap\PR_s(I)\cap\RE_s(I)$ that respects
class precedence from $I$; by \autoref{lem:access_choice},
$\psi_s^C(I)=\{\overline{\CR}_s(I)\}$, and the matching-level axioms
require $\mu^{-1}(s)\in\psi_s^C(\mathcal D_s(\mu))$. DA based on $\CR$
satisfies every required property except possibly class precedence by
\autoref{thm:existence}; its terminal choice $\CR_s(I_s^*)$
(\autoref{cor:da_co_seats}) respects class precedence from $I_s^*$ by
\autoref{lem:access_choice}, hence from $\mathcal D_s(\mu)\subseteq I_s^*$.
Conversely, any mechanism satisfying the axioms yields at each school a
matching for $s$ in
$\NW_s(\mathcal D_s(\mu))\cap\PR_s(\mathcal D_s(\mu))\cap\RE_s(\mathcal D_s(\mu))$
respecting class precedence from $\mathcal D_s(\mu)$ (reserve utilization
by \autoref{lem:mru-demand}, the rest by definition), so
\autoref{lem:access_choice} gives
$\mu^{-1}(s)=\overline{\CR}_s(\mathcal D_s(\mu))$. The restricted
mechanism therefore satisfies the extended choice axioms, and since
$\overline{\CR}_s$ is path independent and size monotone, Theorem~1 of
\citet{Dogan/Imamura/Yenmez:2025} identifies it as DA based on
$\overline{\CR}$, whose school matching is that of DA based on $\CR$.

\section{Proof of \autoref{ord:thm-two}}
\label{ord:proof-two}

We first show that any departure from the $t_1\triangleright t_2$ rule reveals an elementary
conflict resolved through $t_2$. For some reserve problem $(I,\tau)$,
write $X\coloneqq \Psi(I,\tau)$ and $Y\coloneqq X^{1\triangleright2}(I,\tau)$, and
suppose that $A(X)\ne A(Y)$. Both assignments maximize reserve utilization
and hence have the same cardinality. By
\autoref{lem:lexicographic_dominance}, $Y$ is the lexicographically
maximal maximum reserve assignment under the order $t_1\triangleright t_2$.

Choose $i_2\in A(Y)\setminus A(X)$, and let $t\in\{t_1,t_2\}$ be the reserve
through which $i_2$ is assigned under $Y$.

Because $i_2$ is eligible for $t$ but unassigned under $X$, maximum
reserve utilization implies that $t$ is full under $X$. Hence
$|X_t|=r^t\geq |Y_t|$. Since
$i_2\in Y_t\setminus X_t$, there exists
$j\in X_t\setminus Y_t$.

Suppose first that some such $j$ is unassigned under $Y$.
Replacing $i_2$ by $j$ through reserve $t$ gives another feasible maximum
reserve assignment. Since $j$ is assigned through $t$ under $X$ while
$i_2$ is unassigned, priority-respect of $X$ gives $j\succ^t i_2$.
The replacement therefore lexicographically improves $Y$, contradicting
its maximality.

Thus every student in $X_t\setminus Y_t$ is reserve-selected
under $Y$, and hence is assigned there through the other reserve, denoted
$t'$. In particular, $X_t\subseteq A(Y)$. Since
$i_2\in A(Y)\setminus A(X)$ and $|A(X)|=|A(Y)|$, there exists
$d\in A(X)\setminus A(Y)$. The preceding inclusion implies
$d\in X_{t'}$. Choose $j\in X_t\setminus Y_t$; then
$j\in Y_{t'}$.

Consider the assignment obtained from $Y$ by replacing
$(t,i_2),(t',j)$ with $(t,j),(t',d)$.
It is feasible and uses the same number of reserved seats. As above,
priority-respect of $X$ gives $j\succ^t i_2$.

If $t=t_1$, the new assignment improves $Y$ at the first reserve in the
order $t_1\triangleright t_2$, contradicting lexicographic maximality. Hence $t=t_2$ and
$t'=t_1$. Moreover, lexicographic maximality of $Y$ requires
$j\succ^1 d$; otherwise replacing $j$ by $d$ at the earlier reserve $t_1$
would improve $Y$. We have therefore found students $i_2,j,d$ such that
$(t_2,j),(t_1,d)\in X$, $t_2\in\tau(i_2)$, $t_1,t_2\in\tau(j)$,
$j\succ^1 d$, and $j\succ^2 i_2$, where the last inequality follows from
priority-respect of $X$.

Both reserves are full under $X$. Reserve $t_2$ is full because $i_2$ is
$t_2$-eligible and unassigned. If $t_1$ were not full, we could move
$j$ from $t_2$ to $t_1$ and assign $i_2$ through the resulting vacant $t_2$ seat,
increasing reserve utilization, a contradiction.

Let $i_1$ be the lowest-$\succ^1$ student in $X_1$. Since
$d\in X_1$ and $j\succ^1 d$, we have $j\succ^1 i_1$. Define
$D_1\coloneqq X_1\setminus\{i_1\}$ and
$D_2\coloneqq X_2\setminus\{j\}$. Then
$|D_1|=r^1-1$ and $|D_2|=r^2-1$. Every student in $D_1$ has
higher $t_1$-priority than $i_1$, while priority-respect of $X$ implies that
every student in $D_2$ has higher $t_2$-priority than $i_2$.

Now delete every unassigned student other than $i_2$. By irrelevance of rejected students, the
chosen assignment remains $X$. Next, for every remaining selected student
other than $j$, delete all eligibility links except the one used by $X$;
retain both eligibility links of $j$ and only the $t_2$-eligibility of $i_2$.
Each deleted eligibility link is unused, so irrelevance of unused eligibility again leaves the chosen
assignment unchanged.

The resulting problem is an elementary conflict with students $i_1,j,i_2$
and dedicated students $D_1,D_2$. The preserved assignment has
$i_1\to t_1$ and $j\to t_2$, so $\Psi$ resolves this elementary conflict
through $t_2$.

Thus any departure from the $t_1\triangleright t_2$ rule produces an elementary conflict
resolved through $t_2$. By symmetry, any departure from the $t_2\triangleright t_1$ rule
produces an elementary conflict resolved through $t_1$.

If elementary conflicts occur, consistent resolution requires that they all
be resolved through the same reserve. If they are resolved through $t_1$,
the first implication rules out every departure from the $t_1\triangleright t_2$ rule. If
they are resolved through $t_2$, the symmetric implication rules out every
departure from the $t_2\triangleright t_1$ rule. If no elementary conflict occurs, the first
implication rules out every departure from the $t_1\triangleright t_2$ rule.

Hence there exists an order $\triangleright$, either $t_1\triangleright t_2$ or $t_2\triangleright t_1$, such that, for
every reserve problem $(I,\tau)$, $A(\Psi(I,\tau))=A(X^{\triangleright}(I,\tau))$. Whenever elementary conflicts occur, the first reserve in
$\triangleright$ is precisely the reserve through which $\Psi$ assigns
their multi-eligible student.

\section{Proof of \autoref{prop:access}}
\label{supp:access}

Fix a reserve problem $(I,\tau)$, and let
$I^1\coloneqq \{i\in I:t_1\in\tau(i)\}$,
$I^2\coloneqq \{i\in I:t_2\in\tau(i)\}$, and $I^{1\cap 2}\coloneqq I^1\cap I^2$.
Let $n$ denote the maximum feasible
number of reserve assignments at $(I,\tau)$.

We prove the stronger inclusion $X_2^{2\triangleright1}\subseteq A(X^{1\triangleright2})$. The symmetric statement follows by interchanging $t_1$ and $t_2$.

Let $W\coloneqq X_2^{2\triangleright1}$ and $T\coloneqq X_1^{1\triangleright2}$, and write $q_2\coloneqq |W|$ and $q_1\coloneqq |T|$. Because the first class is processed lexicographically subject to preserving maximum reserve utilization, $q_2=\min\{r^2,|I^2|,n\}$ and $q_1=\min\{r^1,|I^1|,n\}$. Under $t_1\triangleright t_2$, after $T$ is fixed, the rule assigns through $t_2$ the $m\coloneqq n-q_1$ highest-$\succ^2$ students in $I^2\setminus T$.

Under $t_2\triangleright t_1$, a $q_2$-set $Q\subseteq I^2$ can be assigned through $t_2$ as part of an $n$-seat reserve assignment if and only if $|Q\cap I^{1\cap 2}|\le d_2\coloneqq |I^1|-(n-q_2)$, because the remaining $n-q_2$ students must be distinct students in $I^1\setminus Q$. Hence $W$ is the $\succ^2$-lexicographically maximal $q_2$-subset of $I^2$ satisfying this bound.

We first show that $|W\setminus T|\le m$. If $q_1+q_2\le n$, this is immediate. Suppose $q_1+q_2>n$. For two reserve categories, $n=\min\{r^1+r^2,\ |I^1\cup I^2|,\ r^1+|I^2|,\ r^2+|I^1|\}$. The inequality $q_1+q_2>n$ rules out each term other than $|I^1\cup I^2|$ as the binding one, so $n=|I^1\cup I^2|$. Therefore $|W\cap T|\ge q_1+q_2-n$, and hence $|W\setminus T|\le n-q_1=m$.

Now fix $i\in W\setminus T$ and let
$H_i\coloneqq \{j\in I^2\setminus(W\cup T):j\succ^2 i\}$.
If $H_i=\varnothing$, then the number of students in $I^2\setminus T$ weakly above $i$ is at most
$|W\setminus T|\le m$,
so $i\in X_2^{1\triangleright2}$.

Suppose instead that $H_i\neq\varnothing$. Every $j\in H_i$ must be dual-eligible: otherwise $j$ could have been retained before $i$ under $t_2\triangleright t_1$ without affecting the possibility of filling $t_1$. Moreover, when such a $j$ was rejected, the bound on $|Q\cap I^{1\cap 2}|$ was already binding. Writing $g\coloneqq |W\cap T|$, this gives
$|H_i|\le |I^{1\cap 2}\setminus T|-d_2+g$.

But $|I^{1\cap 2}\setminus T|-d_2=|I^{1\cap 2}|-|T\cap I^{1\cap 2}|-|I^1|+n-q_2\le n-q_1-q_2$, since $|I^1\setminus I^{1\cap 2}|+|T\cap I^{1\cap 2}|\ge |T|=q_1$. Thus
$|H_i|\le n-q_1-q_2+g$.
The number of students in $I^2\setminus T$ weakly above $i$ is therefore at most $|W\setminus T|+|H_i|\le (q_2-g)+(n-q_1-q_2+g)=n-q_1=m$. Hence $i\in X_2^{1\triangleright2}$.

It follows that $W\setminus T\subseteq X_2^{1\triangleright2}$. Students in $W\cap T$ are already reserve-selected through $t_1$ under $t_1\triangleright t_2$, so $X_2^{2\triangleright1}=W\subseteq A(X^{1\triangleright2})$. Finally, a student outside $I^1$ cannot be routed through $t_1$. Therefore $X_2^{2\triangleright1}\setminus I^1\subseteq X_2^{1\triangleright2}\setminus I^1$. Interchanging $t_1$ and $t_2$ gives the second inclusion.

\ifsupplement
\clearpage
\newgeometry{left=1.20in, top=1.20in, right=1.20in, bottom=1.20in}

\setcounter{section}{0}\setcounter{equation}{0}
\renewcommand{\thesection}{S\arabic{section}}
\renewcommand{\theequation}{S.\arabic{equation}}
\setcounter{table}{0}\renewcommand{\thetable}{S\arabic{table}}\def\theHtable{supp.\arabic{table}}
\setcounter{theorem}{0}\setcounter{lemma}{0}\setcounter{proposition}{0}
\setcounter{corollary}{0}\setcounter{definition}{0}\setcounter{example}{0}\setcounter{remark}{0}
\renewcommand{\thetheorem}{S.\arabic{theorem}}
\renewcommand{\thelemma}{S.\arabic{lemma}}
\renewcommand{\theproposition}{S.\arabic{proposition}}
\renewcommand{\thecorollary}{S.\arabic{corollary}}
\renewcommand{\thedefinition}{S.\arabic{definition}}
\renewcommand{\theexample}{S.\arabic{example}}
\renewcommand{\theremark}{S.\arabic{remark}}

\def\theHsection{supp.\arabic{section}}
\def\theHequation{supp.\arabic{equation}}
\def\theHtheorem{supp.\arabic{theorem}}
\def\theHlemma{supp.\arabic{lemma}}
\def\theHproposition{supp.\arabic{proposition}}
\def\theHcorollary{supp.\arabic{corollary}}
\def\theHdefinition{supp.\arabic{definition}}
\def\theHexample{supp.\arabic{example}}
\def\theHremark{supp.\arabic{remark}}

\begin{center}
{\large\textbf{Supplementary Appendix}}\\[0.3em]
\end{center}

\vspace{1em}



\section{Computation of the Classify-and-Rank Choice Rule}
\label{supp:computation}

The classify-and-rank choice rule is defined through extendability tests, each
asking whether a candidate commitment is compatible with maximum reserve
utilization. We show that these tests, and hence the rule itself, are
computationally tractable.

\begin{proposition}\label{prop:computation}
For every school $s\in\mathcal S$, applicant set $I\subseteq\mathcal I$, and
classify-and-rank specification, $\CR_s(I)$ can be computed in polynomial time.
\end{proposition}

\begin{proof}
For the computational claim, consider the bipartite graph whose left vertices
are the students in $I$, whose right vertices consist of $r_s^t$ slots for each
category $t\in\mathcal T$ with $r_s^t>0$, and whose edges join each student
$i\in I$ to the slots of the categories in $\tau(i)$. Reserve assignments from
$I$ correspond to matchings in this graph, so $n_s(I)$ is computable in
polynomial time as the cardinality of a maximum matching.

The extendability
test at step $\ell$ of the class-commitment stage asks whether some reserve
assignment of cardinality $n_s(I)$ realizes every commitment in
$Z^{\ell-1}\cup\{f_\ell\}$; because the committed students must be covered,
this does not reduce to maximum-cardinality matching alone. Instead, restrict
the edges of each committed student to the slots of the categories of her
committed class, give the restricted students' edges weight $|\mathcal I|+1$
and all remaining edges weight one, and compute a maximum-weight matching.
If $Z^{\ell-1}\cup\{f_\ell\}$ contains two pairs with the same student, no
class commitment contains it and the test fails outright; so assume its pairs
involve $c_\ell$ distinct students, each restricted as above. A matching that
covers all $c_\ell$ restricted students and has total cardinality $n_s(I)$ has
weight $(|\mathcal I|+1)c_\ell+n_s(I)-c_\ell$; since the restricted graph is a
subgraph, no matching has cardinality above $n_s(I)$, so any matching covering
at most $c_\ell-1$ restricted students has weight at most this value minus
$|\mathcal I|$, and no accumulation of unrestricted students can compensate.
Hence the test succeeds if and only if the maximum weight equals the displayed
value.

The tests in the exact-assignment stage are analogous: every student
committed in $Z_s(I)$ is restricted to the slots of the categories of her
class, the students of already-selected pairs and of the pair under
consideration are restricted further to the slots of their own category. If
the pair under consideration and an already-selected pair involve the same
student but different categories, the test fails outright. Otherwise, the
count of restricted students again includes the student of the pair being
tested. It suffices to test containment rather than equality of the induced
class commitment, because $|Z_s(I)|=n_s(I)$ leaves no room for an uncommitted
student. Taking $\min\{r_s^t,|I|\}$ slots per category keeps the graph
polynomial in the input even when capacities are encoded in binary. The first
stage performs at most $|F_s(I)|$ tests and the second at most $|E_s(I)|$, so
the rule requires a polynomial number of maximum-weight bipartite matching
computations, after which the open seats are filled according to
$\succ_s^o$.
\end{proof}

\section{Limits of the Access Comparative Static}
\label{ord:boundary}

The student-by-student access protection underlying
\autoref{prop:access} has two limits: it need not extend from two singleton
reserve classes to arbitrary common-priority classes, and, when the school
also has open seats, it concerns reserved-seat access rather than admission.

\smallskip
\noindent\emph{Multi-category classes.}
Let the categories be $t_1,t_2,t_3$, with one seat in each, grouped into
the classes $B^1=\{t_1,t_2\}$ and $B^2=\{t_3\}$.
There are four students with
$\tau(u)=\{t_2,t_3\}$, $\tau(v)=\{t_1\}$,
$\tau(w)=\{t_1,t_2,t_3\}$, and $\tau(z)=\{t_2\}$.
Class priorities are
$u\succ^{B^1} w\succ^{B^1} z\succ^{B^1} v$ and
$z\succ^{B^2} v\succ^{B^2} u\succ^{B^2} w$.

Maximum reserve utilization is three. With $B^1$ before $B^2$,
classify-and-rank assigns $u,v$ through $B^1$ and $w$ through $B^2$,
realized by $(t_2,u),(t_1,v),(t_3,w)$. With $B^2$ before $B^1$, it assigns
$u$ through $B^2$ and $w,z$ through $B^1$, realized by
$(t_3,u),(t_1,w),(t_2,z)$. Hence $v$ is reserve-selected when $B^1$ is
first but not when $B^2$ is first, so the corresponding class-level access
inclusion fails.

The failure arises from internal routing within a multi-category class.
The order still determines whose priority prevails, but no longer
guarantees student-by-student access inclusion.

\smallskip
\noindent\emph{Open seats.}
\autoref{prop:access} continues to hold when the school also has open seats,
because classify-and-rank fills the reserves before the open seats and its
reserve stage does not depend on them. It compares reserved-seat access,
however, and does not extend to admission. Let $r^1=r^2=1$ with one open
seat, and let the students be $j$, eligible for both categories, $i_1$,
eligible only for $t_1$, and $i_2$ and $i_3$, eligible only for $t_2$, with
priorities $j\succ^1 i_1$, $j\succ^2 i_2\succ^2 i_3$, and
$i_1\succ^o i_3\succ^o i_2$; the position of $j$ under $\succ^o$ is
immaterial. Under $t_1\triangleright t_2$, $j$ takes $t_1$, $i_2$ takes $t_2$,
and the open seat goes to $i_1$, so $i_3$ is rejected. Under
$t_2\triangleright t_1$, $j$ takes $t_2$, $i_1$ takes $t_1$, and the open seat
goes to $i_3$. Thus $i_3$, eligible only for $t_2$, is admitted when $t_2$ is
first but not when $t_1$ is first, although no student eligible only for
$t_2$ holds a reserved seat when $t_2$ is first, so the inclusion of
\autoref{prop:access} holds. The order changes who competes for the open
seat.

\section{Proofs for the Single-Priority Case}
\label{supp:single-priority}

This section proves the three parts of \autoref{cor:single-priority}.

\begin{proof}[Proof of \autoref{cor:single-priority}(i)]
Fix an applicant set $I\subseteq\mathcal I$. Since
$\mathcal B_s=\{\mathcal T\}$, the class-commitment stage applies the
greedy procedure under $\succ_s$: it selects the highest-priority remaining
student whenever doing so remains compatible with attaining $n_s(I)$
utilized reserves. This is the first stage of the meritorious horizontal
choice rule of \citet{sonmez/yenmez:22}. The open-seat stage then fills the
remaining capacity with the highest-priority students not yet selected,
which is its second stage. Hence the two rules select the same student set
from every applicant set $I$.
\end{proof}

\begin{proof}[Proof of \autoref{cor:single-priority}(ii)]
Both conditions require ordinary priority-respect. Fix an applicant set
$I\subseteq\mathcal I$ and a matching $(J,\sigma_s)$
for school $s$ that satisfies ordinary priority-respect from $I$. In the
single-priority environment with $\mathcal B_s=\{\mathcal T\}$, the sole
joint-priority group is $G=\mathcal T\cup\{o\}$ and $A_s^G=J$.
Open-seat assignments do not contribute to reserve utilization, so
$n_s^G(Y)=n_s(Y)$ for every student set $Y$. Conditional on the assumed
priority-respect, \autoref{def:joint-priority} therefore requires exactly
\[
n_s\bigl((J\setminus\{j\})\cup\{i\}\bigr)<n_s(J)
\quad\text{for every }i\in I\setminus J,\ j\in J
\text{ with }i\succ_s j.
\]
This is the no-justified-envy condition of \citet{sonmez/yenmez:22}.
The same argument applies to a school in a full matching by taking $I$ to
consist of its assignees and the students who prefer it to their assignments.
\end{proof}

\begin{proof}[Proof of \autoref{cor:single-priority}(iii)]
By part~(ii), no justified envy is equivalent to the condition displayed
there. Fix an applicant set $I\subseteq\mathcal I$ and write
$c\coloneqq \min\{q_s,|I|\}$ and $r\coloneqq n_s(I)$.

For existence, \autoref{prop:greedy} gives non-wastefulness and maximum
reserve utilization for $\CR_s$. Its selected set has size $c$ and supports
$r$ reserve assignments. In the single-priority environment, that set is
the greedy basis of the reserve transversal matroid, topped up to size $c$
in priority order. By \autoref{lem:greedy_topup}, replacing a selected
student with a higher-priority excluded student reduces its reserve rank.
This is precisely the no-justified-envy condition in part~(ii).

We next establish uniqueness. Let $J\neq J'$ be two student sets selected
from $I$, each satisfying non-wastefulness and no justified envy and supporting
maximum reserve utilization. Non-wastefulness gives $|J|=|J'|=c$, and the
support condition gives $n_s(J)=n_s(J')=n_s(I)=r$.
If $|I|\leq q_s$, non-wastefulness gives $J=J'=I$, a contradiction. Hence
$|I|>q_s$.

Let $\mathcal N$ be the transversal matroid on $I$ presented by
$r_s^t$ copies of $I^t$ for each $t\in\mathcal T$, so its rank function
is $n_s$. Let $\mathcal N^*$ be its dual. Since $\mathcal N$ has rank $r$, a set
$Y\subseteq I$ satisfies $n_s(Y)=r$ if and only if $I\setminus Y$ is
independent in $\mathcal N^*$. Consequently,
$Z\coloneqq I\setminus J$ and $Z'\coloneqq I\setminus J'$
are distinct independent sets of $\mathcal N^*$ with
$|Z|=|Z'|=|I|-c$.

Let $i$ be the $\succ_s$-highest student in
$J\mathbin{\triangle}J'$. Without loss of generality,
$i\in J\setminus J'=Z'\setminus Z$.
Since $Z'\setminus\{i\}$ and $Z$ are independent in $\mathcal N^*$ and
$|Z'\setminus\{i\}|<|Z|$,
the augmentation axiom yields a student
$j\in Z\setminus Z'=J'\setminus J$
such that $(Z'\setminus\{i\})\cup\{j\}$
is independent in $\mathcal N^*$. Taking complements,
$Y\coloneqq (J'\setminus\{j\})\cup\{i\}$
satisfies $|Y|=c$ and $n_s(Y)=r=n_s(J')$.
Moreover, $i\succ_s j$ by the choice of $i$, while
$i\notin J'$ and $j\in J'$. This contradicts no justified envy for $J'$.

Thus at most one restricted choice from each applicant set satisfies
non-wastefulness and no justified envy and supports maximum reserve
utilization. By the
existence argument and part~(i), this unique restricted choice rule is the
meritorious horizontal choice rule of \citet{sonmez/yenmez:22}.
\end{proof}

\section{More on Gale Dominance}
\label{supp:gale-path}

Path independence is equivalent to substitutability and
consistency \citep{AizMal:1981AutCont}. The next result shows that choosing a
Gale-dominant matching whenever one exists need not preserve this property.

\begin{proposition}\label{prop:PI}
There exist an environment and a school $s\in\mathcal S$ with the following
property. Let $C_s$ be any choice rule that, for every applicant set
$I\subseteq\mathcal I$, selects a matching from
$\RE_s(I)\cap\PR_s(I)$ and selects a Gale-dominant matching whenever one
exists. Then the restricted choice rule $\overline C_s$ induced by $C_s$
violates path independence.
\end{proposition}

The result is existential: in the single-priority environment,
classify-and-rank satisfies the premise and is path independent
(\autoref{path_independence}).

\begin{proof}
Consider a single school $s$ of the form constructed in the proof of
\autoref{prop_nonexistence_gale}, with $P=\{i_3,i_6\}$, $Q=\{i_4,i_5\}$,
reserve priorities $i_2\succ_s^1 i_1\succ_s^1 i_6\succ_s^1 i_3$ and
$i_2\succ_s^2 i_1\succ_s^2 i_4\succ_s^2 i_5$, open-seat priority
$i_1\succ_s^o\cdots\succ_s^o i_6$, and student universe
$\mathcal I=\{i_1,\ldots,i_6\}$. By Step~1 of that proof, for the full
applicant set $\mathcal I$, the matchings in
$\RE_s(\mathcal I)\cap\PR_s(\mathcal I)$ have selected student sets
\[
X=\{i_1,i_2,i_3,i_6\},
\qquad
Y=\{i_1,i_2,i_4,i_6\},
\qquad
Z=\{i_1,i_2,i_4,i_5\},
\]
none of which is Gale-dominant: $X$ and $Z$ both dominate $Y$, but $X$ has
the higher third-ranked member and $Z$ the higher fourth-ranked member.
Hence
$\overline C_s(\mathcal I)\in\{X,Y,Z\}$, and we consider the three
possibilities. We use twice the fact that
$\RE_s(\mathcal I\setminus\{i_3\})\cap\PR_s(\mathcal I\setminus\{i_3\})$
has selected student sets $Y$ and $Z$, of which $Z$ is Gale-dominant, since the
two sets agree except that $Z$ contains $i_5$ where $Y$ contains $i_6$ and
$i_5\succ_s^o i_6$. Consequently
$\overline C_s(\mathcal I\setminus\{i_3\})=Z$.

\begin{enumerate}
\item
Suppose $\overline C_s(\mathcal I)=X$. Then $i_6\in\overline C_s(\mathcal I)$
and $i_6\in\mathcal I\setminus\{i_3\}$, but
$\overline C_s(\mathcal I\setminus\{i_3\})=Z$ does not contain $i_6$.
Substitutability is violated.

\item
Suppose $\overline C_s(\mathcal I)=Y$. Since $i_3\notin Y$ and
$\overline C_s(\mathcal I\setminus\{i_3\})=Z\neq Y$, consistency is violated.

\item
Suppose $\overline C_s(\mathcal I)=Z$. Deleting $i_6$ instead, the selected
student sets of
$\RE_s(\mathcal I\setminus\{i_6\})\cap\PR_s(\mathcal I\setminus\{i_6\})$ are
$Z$ and $W\coloneqq \{i_1,i_2,i_3,i_4\}$, and $W$ is Gale-dominant because
$i_3\succ_s^o i_4$ and $i_4\succ_s^o i_5$. Hence
$\overline C_s(\mathcal I\setminus\{i_6\})=W$. Since $i_6\notin Z$ and
$W\neq Z$, consistency is violated.
\end{enumerate}
\end{proof}

The example also establishes failure of path independence for the
\emph{reverse-rejecting rule} of \citet{aziz2024efficient} when its baseline
ranking is fixed at $\succ_s^o$. In this environment all four seats can be
filled through reserves in each of the three pools used above. Reverse
rejection selects $Z$ from $\mathcal I$ and $W$ after deleting the rejected
student $i_6$, violating consistency. The lexicographic maximizer based on
$\succ_s^o$ selects $X$ from $\mathcal I$ and $Z$ after deleting $i_3$;
the retained student $i_6$ then loses admission, violating substitutability.

\section{Independence of the Characterizing Properties in \autoref{thm:characterization}}
\label{supp:independence}

For each property, we give an environment and a mechanism that satisfies
the other five properties at every preference profile but violates that
property at some profile. In the one-school examples, write
$A(R)\coloneqq \{i\in\mathcal I:s\mathrel{P_i}\varnothing\}$ for the students who report
$s$ acceptable. Unspecified priority comparisons may be completed
arbitrarily.

\emph{Individual rationality.}
There is one student $a$ and one school $s$ with capacity one and one
reserved seat $t$, for which $a$ is eligible. Assign $a$ to $s$ through
$t$ at every profile. Individual rationality fails when
$\varnothing\mathrel{P_a}s$. The school is always full and the reserved seat
always used, so non-wastefulness and maximum reserve utilization hold.
Priority-respect and class precedence are vacuous, and the school
assignment is constant, so the mechanism is strategy-proof.

\emph{Non-wastefulness.}
There are two students $h,\ell$, one school with capacity two, and one
reserved seat $t$. Both students are eligible and
$h\succ_s^t\ell$. Whenever $A(R)$ is nonempty, assign through $t$ the
highest-$\succ_s^t$ student in $A(R)$ and assign nobody else. If both
students find $s$ acceptable, $\ell$ is left unmatched despite a vacant
seat, so non-wastefulness fails. Individual rationality and
priority-respect hold, and the reserve is used whenever feasible, so
reserve utilization is maximal. The rule is strategy-proof. Class
precedence holds because the student assigned through the unique class is
the highest-priority student in the relevant demand set.

\emph{Priority-respect.}
There are three students $d,h,\ell$, one school with capacity two, and
one reserved seat $t$, for which only $d$ is eligible. Let
$h\succ_s^o\ell$. If $d\in A(R)$, admit $d$ through $t$ and fill the
remaining seat with $\ell$ if possible, using $h$ only if $\ell$ is
unavailable. If $d\notin A(R)$, admit all acceptable students among
$h,\ell$ openly. When all three students find $s$ acceptable, $\ell$ is
admitted openly while $h$ is unmatched, so priority-respect fails.
Individual rationality, non-wastefulness, and maximum reserve utilization
hold. The rule is strategy-proof: $d$ and $\ell$ are assigned exactly
when they report $s$ acceptable, while $h$ cannot obtain a seat by
changing her report when both $d$ and $\ell$ are acceptable. Class
precedence holds because $d$ is the only student eligible for the reserve
class and is assigned through it whenever she belongs to the relevant
demand set.

\emph{Maximum reserve utilization.}
There are three students $x,y,h$, one school with capacity two, and singleton
reserve classes $B_s^1=\{t_1\}$ and $B_s^2=\{t_2\}$, each with capacity one,
ordered $t_1\triangleright_s t_2$. Student $x$ is eligible for both reserves,
$y$ only for $t_1$, and $h$ for neither. Let
$x\succ_s^1y$ and $x\succ_s^o h\succ_s^o y$.
At every profile, admit the two highest students in the fixed order
$x\succ h\succ y$ among those who find $s$ acceptable. If $x$ is admitted,
assign her through $t_1$; otherwise assign $y$ through $t_1$ if she is
admitted. Assign any remaining admitted student openly. Thus $t_2$ is never
used.

When all three students find $s$ acceptable, the mechanism admits $x$ and
$h$, assigning $x$ through $t_1$ and $h$ openly, while $y$ is unassigned.
Yet the alternative assigning $y$ through $t_1$ and $x$ through $t_2$
admits $y$ and uses two reserved seats. Maximum reserve utilization therefore
fails. Individual rationality and non-wastefulness are immediate.
Priority-respect holds because $x\succ_s^1y$ and
$x\succ_s^o h\succ_s^o y$. The admission rule is responsive to a fixed
ranking and hence strategy-proof. Class precedence also holds. At the first class, condition~(iii) of
\autoref{def:class_precedence} requires every alternative to keep $x$, who
outranks $y$ under $\succ_s^1$, in that class, which has a single seat, so
$y$ cannot be admitted through it. At the second class, condition~(ii)
keeps $x$ in the first class, and neither $y$ nor $h$ is eligible for $t_2$.

\emph{Strategy-proofness.}
There are three students $a,b,c$ and three schools $x,y,z$, each with
capacity one and one reserved seat in a common category $t$. Every
student is eligible. Reserve and open priorities coincide, with
$x:a\succ b\succ c$, $y:a\succ b\succ c$, and
$z:b\succ a\succ c$. Use school-proposing deferred acceptance, assigning
every matched student through $t$.

The resulting matching is stable with respect to these priorities, and
therefore individually rational, non-wasteful, and priority-respecting.
A filled school uses its reserved seat; if a school is vacant, no student
strictly prefers it to her assignment. Hence maximum reserve utilization
holds. Class precedence follows from
stability: if $s$ is assigned to $i$, every other student in
$\mathcal D_s(\mu)$ has lower $s$-priority than $i$, and therefore cannot replace
$i$ through the unique reserve class without violating condition~(iii).

The mechanism is nevertheless manipulable. Suppose the true preferences
are $a:z\mathrel{P_a}x\mathrel{P_a}y\mathrel{P_a}\varnothing$,
$b:x\mathrel{P_b}z\mathrel{P_b}y\mathrel{P_b}\varnothing$, and
$c:x\mathrel{P_c}y\mathrel{P_c}z\mathrel{P_c}\varnothing$. School-proposing DA assigns
$(a,x),(b,z),(c,y)$. If $a$ instead reports
$zP'_ayP'_axP'_a\varnothing$, it assigns
$(a,z),(b,x),(c,y)$. Thus $a$ profitably manipulates.

\emph{Class precedence.}
There are three students $a,b,c$, one school with capacity two, and
singleton reserve classes $B_s^1=\{t_1\}$ and $B_s^2=\{t_2\}$, each with
capacity one. Student $a$ is eligible for both reserves, $b$ only for
$t_1$, and $c$ only for $t_2$, with
$a\succ_s^1b$ and $a\succ_s^2c$. Fix the specification
$t_1\triangleright_s t_2$, but consider the DA mechanism based on
classify-and-rank under the reverse order
$t_2\triangleright_s t_1$.

By \autoref{thm:existence}, this mechanism satisfies individual
rationality, non-wastefulness, priority-respect, maximum reserve
utilization, and strategy-proofness at every profile. Suppose all three
students find $s$ acceptable. The reverse-order mechanism assigns $a$
through $t_2$ and $b$ through $t_1$, leaving $c$ unmatched. Consider
instead the alternative assigning $a$ through $t_1$ and $c$ through
$t_2$. It uses the same number of reserved seats and changes the admitted
set. At the first class there is no earlier class to preserve, no student
assigned there under the original matching has higher $t_1$-priority than
$a$, and the alternative assigns $a$ through that class. Hence the
alternative witnesses a violation of class precedence.

\section{Proof of \autoref{prop:sibling-adjusted-outcome}}
\label{supp:sa-da-pi}\label{supp:sa-da-allocation}

We first establish properties of the SA-DA choice rules and then prove
the proposition. Fix a program and suppress its index. Reserves $\exc$ and $\dis$ form
singleton classes in a fixed order. Let $\succ$ be the open-seat priority
and $\Sigma$ the students with sibling priority. Assume that every member
of $\Sigma$ precedes every nonmember under $\succ$. All priority rankings
and sibling statuses are fixed across applicant sets.

For reserve capacities $r=(r^{\exc},r^{\dis})$, open capacity $o$, and menu $N$,
write $\mathcal R^r(N)$ for the students selected through reserves by
classify-and-rank, and $C^{r,o}(N)$ for its complete selected set. Put
$Q\coloneqq r^{\exc}+r^{\dis}+o$. By the construction of classify-and-rank,
\begin{equation}\label{eq:sa-pi-open}
 C^{r,o}(N)=\mathcal R^r(N)\cup
 \operatorname{top}_{Q-|\mathcal R^r(N)|}^{\succ}(N\setminus \mathcal R^r(N)).
\end{equation}
For fixed capacities, classify-and-rank is path independent and
substitutable, selecting $\min\{Q,|N|\}$ students. Here
$\operatorname{top}_{m}^{\succ}$ selects up to $m$ highest-priority students.

\begin{lemma}[One-seat comparisons]\label{lem:sa-pi-exchange}
For a fixed menu $N$:
\begin{enumerate}[label=(\roman*)]
\item Removing one unit of $\exc$- or $\dis$-capacity removes at most one student
from the reserve-selected set, and adds none. Consequently, reducing any
one of $r^{\exc},r^{\dis},o$ by one weakly shrinks $C^{r,o}(N)$.
\item Fix $\mathbf c=(a,d,o)$ and put
$C_k(N)\coloneqq C^{(a,d-k),o+k}(N)$, $0\le k\le d$.
Either $C_{k+1}(N)=C_k(N)$ or
\[
 C_{k+1}(N)=(C_k(N)\setminus\{h\})\cup\{y\},
\]
where $h$ was reserve-selected at $k$, and $y$ is the
$\succ$-highest student excluded at $k$, with $y\succ h$.
\end{enumerate}
\end{lemma}
\begin{proof}
For the reserve-stage claim, represent reserve assignments as bipartite
matchings of students to labelled reserve slots, choosing a labelled slot
$v$ removed when capacity is reduced. Choose representations agreeing on
common student--category pairs, with the reduced-capacity assignment
avoiding $v$. Under singleton classes, each selected reserve assignment
is a lexicographically maximal maximum-cardinality matching, ordered by
category precedence and priority within category
(\autoref{lem:lexicographic_dominance}). In the symmetric difference of
the two assignments, an alternating component not containing $v$ would
either augment one of the matchings or permit a cardinality-preserving
lexicographic improvement of one of them. Thus the symmetric difference
is empty or consists of a single alternating path from $v$. Every student
selected under the reduced capacity was already reserve-selected, and at
most one reserve-selected student disappears.

Now use \eqref{eq:sa-pi-open}. Removing an open seat only reduces the
open top-up. Removing a reserve seat either leaves the reserve-selected
set unchanged, decreasing the top-up quota by one, or removes a single
reserve-selected student $h$, leaving the top-up quota unchanged. In both
cases the complete selected set weakly shrinks. Finally, replacing one
$\dis$ seat by one open seat leaves $Q$ constant. If $h$ drops from the
reserve stage, the open top-up increases by one and selects the higher
of $h$ and the best previously excluded student. Hence either the
selected set is unchanged or the asserted improving exchange occurs.
\end{proof}

Fix $\mathbf c=(a,d,o)$, $q\coloneqq a+d+o$, and define
\[
 \mathcal K(N)\coloneqq\{k:\Sigma\cap N\subseteq C_k(N)\},\qquad
 \Phi^{\mathbf c}(N)\coloneqq C_{k(N)}(N),
\]
where
\[
 k(N)\coloneqq\begin{cases}\min\mathcal K(N),&\mathcal K(N)\ne\varnothing,\\
 d,&\mathcal K(N)=\varnothing.\end{cases}
\]
Thus $k(N)$ is the release level of \bodysecref{sec:sibling-adjusted-da}, and
$\Phi^{\mathbf c}$ is the rule applied to the residual students when
$\mathbf c$ gives the $\exc$, $\dis$, and open capacities remaining after the
returning-student and PIE steps.
By \autoref{lem:sa-pi-exchange}(ii), $\mathcal K(N)$ is upward closed: if all siblings are chosen at $k$,
then an excluded entrant at the next release cannot be a sibling; since
$y\succ h$ and all siblings rank above non-siblings, $h$ cannot be a
sibling either.

\begin{lemma}[Residual substitutability]\label{lem:sa-pi-phi}
For every $\mathbf c$, $\Phi^{\mathbf c}$ is substitutable and path independent, and
$|\Phi^{\mathbf c}(N)|=\min\{q,|N|\}$.
\end{lemma}
\begin{proof}
The cardinality follows from fixed-capacity classify-and-rank. It is
enough to check substitutability under deletion of one student $x$.
Let $N'=N\setminus\{x\}$; the claim is immediate if $|N|\le q$.
Write $k\coloneqq k(N)$, $k'\coloneqq k(N')$. By substitutability of each $C_j$,
$k'\le k$: if $\mathcal K(N)$ is nonempty, any release accommodating
siblings in $N$ also accommodates those in $N'$; if it is empty,
$k=d$. When $k'=k$, fixed-level substitutability gives the claim.

Suppose $k'<k$. For each $j=k',\ldots,k-1$, at least one sibling is
excluded from $C_j(N)$, whereas $C_j(N')$ accommodates every sibling
of $N'$. All these sets have size $q$. Substitutability and consistency
of $C_j$ imply that exactly one sibling $e_j$ is excluded from $C_j(N)$,
and
\begin{equation}\label{eq:sa-pi-deletion}
 C_j(N')=\begin{cases}
 (C_j(N)\setminus\{x\})\cup\{e_j\},&x\in C_j(N),\\
 C_j(N),&x\notin C_j(N),\quad(e_j=x).
 \end{cases}
\end{equation}
By sibling precedence $e_j$ is the highest-ranked excluded student.
Comparing consecutive levels using \autoref{lem:sa-pi-exchange}(ii)
shows that $C_{j+1}(N')=C_j(N')$ for $j<k-1$:
if $C_{j+1}(N)=C_j(N)$ this follows immediately from
\eqref{eq:sa-pi-deletion}; otherwise
$C_{j+1}(N)=(C_j(N)\setminus\{h\})\cup\{e_j\}$, and the
unique excluded sibling at $j+1$ must be $h$. Inserting this into
\eqref{eq:sa-pi-deletion} at both levels gives the same set, whether
$x=h$, $x\notin C_j(N)$, or $x\in C_j(N)\setminus\{h\}$.
Therefore $\Phi^{\mathbf c}(N')=C_{k-1}(N')$.

By the one-seat comparison,
$C_k(N)\subseteq C_{k-1}(N)\cup\{e_{k-1}\}$; by
\eqref{eq:sa-pi-deletion},
\[
 C_k(N)\setminus\{x\}\subseteq C_{k-1}(N')=\Phi^{\mathbf c}(N').
\]
This proves substitutability. Because the rule selects
$\min\{q,|N|\}$ students, substitutability implies consistency:
if $\Phi^{\mathbf c}(N)\subseteq M\subseteq N$, then substitutability and equal
choice cardinalities yield $\Phi^{\mathbf c}(M)=\Phi^{\mathbf c}(N)$. These properties imply path independence.
\end{proof}

\begin{lemma}[Residual capacity monotonicity]\label{lem:sa-pi-capacity}
If $\mathbf c'\le \mathbf c$ coordinatewise, then $\Phi^{\mathbf c'}(N)\subseteq\Phi^{\mathbf c}(N)$.
\end{lemma}
\begin{proof}
It suffices to reduce one capacity coordinate by one. Write $U_j\coloneqq C^{\mathbf c'}_j(N)$ and $V_j\coloneqq C^{\mathbf c}_j(N)$ at common release levels. By
\autoref{lem:sa-pi-exchange}(i), $U_j\subseteq V_j$. If $|N|<q$,
the conclusion is immediate. Otherwise $|V_j|=q$,
$|U_j|=q-1$, and $V_j=U_j\cup\{g_j\}$. Let $k,k'$ be the two
release thresholds. If the smaller-capacity rule accommodates siblings
at a level, so does the larger one. Thus either $k\le k'$, or the only
exception occurs when the $\dis$ capacity was reduced and $k=d$, $k'=d-1$.
In that exception, $U_{d-1}\subseteq V_d$ directly, as these are
fixed-rule choices with the same zero $\dis$ capacity and one fewer open seat.

When $k=k'$, use $U_k\subseteq V_k$. If $k<k'$, then for each
$j=k,\ldots,k'-1$, $V_j$ accommodates every sibling but $U_j$ does not.
Thus $g_j$ is its unique excluded sibling, which is the highest-priority
student outside $U_j$. By \autoref{lem:sa-pi-exchange}(ii),
$U_{j+1}$ either equals $U_j$ or exchanges some $h$ for $g_j$.
For $j+1<k'$, in the latter event $g_{j+1}=h$, and in the former
$g_{j+1}=g_j$. In both events $V_{j+1}=U_{j+1}\cup\{g_{j+1}\}=V_j$.
Therefore $V_{k'-1}=V_k$, and one more application of the exchange
lemma gives
$U_{k'}\subseteq U_{k'-1}\cup\{g_{k'-1}\}=V_k$.
Iterate single-coordinate reductions to conclude.
\end{proof}

\paragraph{PIE and returning students.}
Let $I^{\mathrm{ret}}$ be the returning students. Fix a menu-independent total
order on $I^{\mathrm{ret}}$ and, for each returning student, a fixed list of
\textbf{physical} seat types, that is, the $\exc$, $\dis$, open, and PIE positions
counted in the program's reported capacity, as opposed to fallback
positions. Given $E\subseteq I^{\mathrm{ret}}$, process the members of $E$ in this
order, assigning each to the first type on her list with an unoccupied
position, or to a personal, nontransferable fallback position if none is
available. Let $\kappa(E)=(a,d,o,p)$ be the remaining capacities of the
four physical seat types, and write $|\kappa(E)|\coloneqq a+d+o+p$.

\begin{lemma}[Returning-student routing capacities]\label{lem:incumbent-routing}
If $E\subseteq F$, then $\kappa(F)\le\kappa(E)$ coordinatewise. If
$F=E\cup\{i\}$, then
\[
 0\le |\kappa(E)|-|\kappa(F)|\le 1.
\]
\end{lemma}
\begin{proof}
It suffices to insert one returning student $i$ into the fixed processing order.
Before $i$ is processed, both runs have identical remaining capacities.
In the enlarged run, $i$ either takes one available physical position or
uses fallback. Hence, immediately after $i$, the enlarged run has the same
remaining capacities as the original run or has one fewer position of one
type. We claim that this relation persists as the remaining returning students are
processed in their common order.

Indeed, suppose just before a common returning student $j$ the remaining capacity
vectors are equal, or the enlarged run has exactly one fewer position of
some type $t$. If the vectors are equal, $j$ makes the same choice in both
runs. Otherwise, $j$ either selects the same available type in both runs
(in which case the one-position difference persists), or chooses $t$ in the
original run when no $t$-position is available in the enlarged run. In the
latter case, $j$ takes a later available type in the enlarged run, moving
the one-position deficit to that type, or uses fallback, eliminating the
deficit. In each case, all coordinates of the enlarged residual vector
remain weakly smaller and its total differs by zero or one. Induction over
all subsequent incumbents proves both claims.
\end{proof}

The subsequent PIE step selects the top $p$ non-returning students
according to the fixed order of the PIE procedure, in which PIE-eligible
students precede all others, and releases unused PIE capacity to open
positions. For menu $M$ of non-returning students, put
\[
 P_p(M)\coloneqq\operatorname{top}^{\mathrm{PIE}}_p(M),
 \qquad
 G_\kappa(M)\coloneqq P_p(M)\cup
 \Phi^{(a,d,o+p-|P_p(M)|)}(M\setminus P_p(M)).
\]
\begin{lemma}[PIE composition]\label{lem:sa-pi-G}
$G_\kappa$ is path independent, chooses $\min\{|\kappa|,|M|\}$
students, and is monotone in capacities:
$\kappa'\le\kappa$ implies $G_{\kappa'}(M)\subseteq G_\kappa(M)$.
\end{lemma}
\begin{proof}
The size identity is immediate. For capacity monotonicity, note that
$P_{p'}(M)\subseteq P_p(M)$ when $p'\le p$. Any student chosen by the
smaller-capacity residual rule is either in $P_p(M)$ or is chosen by
$\Phi$ from the smaller residual menu $M\setminus P_p(M)$, by
substitutability of $\Phi$. The corresponding residual capacities are
coordinatewise smaller because the number of released PIE slots is
$\max\{0,p-|M|\}$. Apply
\autoref{lem:sa-pi-capacity}.

For path independence, set $U\coloneqq M\cup M'$, $P^*\coloneqq P_p(U)$, and $\mathbf c^*\coloneqq(a,d,o+p-|P^*|)$. The top-$p$ property gives
$P^*\subseteq G_\kappa(M)\cup M'$ and
$P_p(G_\kappa(M)\cup M')=P^*$, so both choices have the same PIE set.
Furthermore $\mathbf c^*\le \mathbf c(M)\coloneqq (a,d,o+p-|P_p(M)|)$. We claim
\[
 \Phi^{\mathbf c^*}(M\setminus P^*)\subseteq G_\kappa(M)\setminus P^*.
\]
Students in $P_p(M)\setminus P^*$ are already in $G_\kappa(M)$;
for all other students the claim follows by residual substitutability
and capacity monotonicity, applied to $M\setminus P_p(M)$.
Residual consistency now permits replacing $M\setminus P^*$ by
$G_\kappa(M)\setminus P^*$, and residual path independence permits
adjoining $M'\setminus P^*$. The resulting equality is
$G_\kappa(M\cup M')=G_\kappa(G_\kappa(M)\cup M')$.
\end{proof}

\paragraph{Proof of the proposition.}
\begin{proof}[Proof of \autoref{prop:sibling-adjusted-outcome}]
\emph{Strategy-proofness.} For applicant sets $J,J'$, put $E\coloneqq J\cap I^{\mathrm{ret}}$, $F\coloneqq(J\cup J')\cap I^{\mathrm{ret}}$, and $N\coloneqq J\setminus I^{\mathrm{ret}}$. Both sides of the path-independence identity
contain exactly $F$. Because $\kappa(F)\le\kappa(E)$, capacity
monotonicity yields
\[
 G_{\kappa(F)}(N)\subseteq G_{\kappa(E)}(N)\subseteq N.
\]
Consistency of $G_{\kappa(F)}$ permits replacing $N$ by
$G_{\kappa(E)}(N)$; its path independence then gives the identity
after adjoining $J'\setminus I^{\mathrm{ret}}$. Hence
$\overline C^{\mathrm{SA}}$ is path independent. Moreover,
\[
 |\overline C^{\mathrm{SA}}(J)|=|E|+\min\{|\kappa(E)|,|J\setminus I^{\mathrm{ret}}|\}.
\]
Adding a non-returning student cannot reduce the selected set size;
adding a returning student increases the first term by one and decreases
residual capacity by at most one. Hence the rule is size monotone.
Substitutability and size monotonicity imply strategy-proofness of the
associated cumulative-offer mechanism. Path independence implies the
standard deferred-acceptance/cumulative-offer equivalence, independence
of the order in which proposals are processed, and that the students
held at termination are $\overline C_s^{\mathrm{SA}}(I_s^*)$. These
conclusions are within-grade statements, taking sibling-priority
statuses and adjustments to submitted preference lists as fixed.

For the remaining properties, fix an outcome and a program $s$, and
suppress $s$. Let $E^*$ and $P^*$ denote the returning and PIE-selected students in $I^*$, $N^*\coloneqq I^*\setminus(E^*\cup P^*)$ the residual students, $(a,d,o)$ the $\exc$, $\dis$, and open capacities remaining after the returning-student and PIE steps, $q\coloneqq a+d+o$, and, with $\mathbf c=(a,d,o)$, $\mathcal K^*\coloneqq\mathcal K(N^*)$, $k\coloneqq k(N^*)$, $R^*\coloneqq\mathcal R^{(a,d-k)}(N^*)$, $\Sigma^*\coloneqq N^*\cap\Sigma$, and $\Phi\coloneqq\Phi^{\mathbf c}(N^*)$. The assignees at $s$ are
$E^*\cup P^*\cup\Phi$, and the residual students holding $\exc$ or $\dis$
positions through classify-and-rank are exactly $R^*$, all eligible.
Any student who prefers $s$ to her assignment proposed to $s$ and
belongs to $I^*$; if she is not assigned to $s$, she belongs to
$N^*\setminus\Phi$.

\emph{Secured enrollment.} A returning student who proposes to her
current program is retained regardless of subsequent proposals. She
therefore remains assigned there unless she is assigned to a
higher-ranked program.

\emph{Sibling accommodation.} At release level $j$, the open stage selects $q-|\mathcal R^{(a,d-j)}(N^*)|$ students outside $\mathcal R^{(a,d-j)}(N^*)$, with siblings ranked first. Thus $j\in\mathcal K^*$, that is, all siblings in $N^*$ are selected, if and only if $|\mathcal R^{(a,d-j)}(N^*)\cup\Sigma^*|\le q$. Whenever $\mathcal K^*\neq\varnothing$, every
sibling in $N^*$ is therefore assigned to $s$, and $k=\min\mathcal K^*$
is the minimum number of released positions doing so. Siblings among
returning and PIE-selected students are assigned to $s$ in any case.

\emph{Structure under release.} Suppose $k\ge1$. If
$k\in\mathcal K^*$, then $k-1\notin\mathcal K^*$, so
$|\mathcal R^{(a,d-k+1)}(N^*)\cup\Sigma^*|>q$. By
\autoref{lem:sa-pi-exchange}(i), $R^*$ is obtained from
$\mathcal R^{(a,d-k+1)}(N^*)$ by removing at most one student, so
$|R^*\cup\Sigma^*|=q$, and the open stage selects exactly
$\Sigma^*\setminus R^*$. If $\mathcal K^*=\varnothing$, then $k=d$ and
$|R^*\cup\Sigma^*|>q$, so the open stage selects $q-|R^*|$ siblings. In both
cases
\begin{equation}\label{eq:sa-structure}
 \Phi\setminus R^*\subseteq\Sigma^*\quad\text{and}\quad|\Phi|=q.
\end{equation}

\emph{Priority-respect.} Let $j$ prefer $s$ to her assignment, and let
$i$ be assigned to $s$ in seat type $t$. A returning holder ranks first
for every seat type at her current program, and a personal fallback
position admits no other student. A PIE holder belongs to the top $p$
students under the PIE order among the non-returning students in
$I^*$, which include $j$, so she precedes $j$. For residual holders of
$\exc$, $\dis$, and open positions selected by classify-and-rank, including
$\dis^r$ holders, priority-respect follows from \autoref{prop:greedy}
applied to $N^*$ at capacities $(a,d-k,o+k)$, since $j\in N^*$. It
remains to consider a disadvantaged student $i$ holding a released
position as a $\dis$ position. Then $k\ge1$, and by
\eqref{eq:sa-structure}, $i$ is a sibling selected in the open stage.
If $j$ is not a sibling, $i$ precedes $j$ under the disadvantaged-reserve
priority, which ranks siblings first. If $j$ is a sibling, $i$ precedes
$j$ under the open-seat priority, which orders siblings in the same way
as the disadvantaged-reserve priority.

\emph{Protection-respecting maximum reserve utilization.} Let $\delta$
denote the number of disadvantaged students in $\Phi\setminus R^*$, which
is zero when $k=0$. Since released positions are assigned first to
disadvantaged students, the residual reserve utilization of the
outcome is $|R^*|+\min\{k,\delta\}$. Consider a matching
$(I_s',\sigma_s')$ satisfying the constraints of
\autoref{def:immediate-violation}. Its members belong to $I^*$, and
students in $E^*\cup P^*$ keep their positions, so
$B\coloneqq I_s'\setminus(E^*\cup P^*)$ satisfies $B\subseteq N^*$ and $|B|\le q$,
with $a$ positions in $\exc$ and $d$ in $\dis$ available. Let $T\subseteq B$
denote the students holding reserve positions for which they are
eligible under $\sigma_s'$. It suffices to show
$|T|\le|R^*|+\min\{k,\delta\}$.

If $k=0$, $T\subseteq N^*$ is matched at capacities $(a,d)$, and $R^*$
maximizes reserve utilization over $N^*$ at these capacities
(\autoref{prop:greedy}); hence $|T|\le|R^*|$.

Let $k\ge1$. Since $R^*$ maximizes reserve utilization over $N^*$ at
$(a,d-k)$ and adding $k$ positions raises this maximum by at most $k$,
$|T|\le|R^*|+k$. By \eqref{eq:sa-structure}, $\Phi\setminus R^*$ consists
of $q-|R^*|$ admitted siblings, all of whom belong to $B$. Hence at most
$|R^*|-|R^*\cap\Sigma^*|$ members of $B$ are not admitted siblings. Admitted
siblings in $R^*$ contribute at most $|R^*\cap\Sigma^*|$ to $T$. Those
outside $R^*$ do not hold $\exc$ positions under $\sigma_s$, so they cannot
move into $\exc$, and they count in $\dis$ only if disadvantaged,
contributing at most $\delta$. Thus $|T|\le|R^*|+\delta$.
\end{proof}

\section{Additional Results on the Chilean Application}
\label{supp:empirical-details}

This section provides implementation and replication details for the empirical analysis in \bodysecref{sec:empirical}, and additional results on exposure, deployed reserve utilization, and the reserve order under SA-DA.

\subsection{Institutional Implementation and Replication}
\label{supp:empirical-implementation}

Throughout this section, the \textbf{focal} reserves are the two reserves
studied in the main text, $\exc$ and $\dis$; PIE positions are not focal.

Admissions in the SAE are conducted separately by grade, beginning with the
twelfth grade and proceeding downward so that sibling priorities can
incorporate higher-grade assignments. Returning students have secured
enrollment at their current program. In the deployed mechanism, they receive
highest priority at that program and, when necessary, a personal,
nontransferable fallback assignment that guarantees continued enrollment.
Returning students at their current program consume ordinary program capacity
whenever capacity permits; personal fallback capacity is used only when the
number of protected incumbents exceeds published program capacity.

The special-needs (PIE) procedure is institutionally distinct from the
disadvantaged and academic-excellence reserves. PIE contracts retain their
published capacity, application-specific position, and the priority order
used in the reconstruction. At programs with a focal reserve, PIE and the non-PIE component
share the program's published total capacity. Students held through PIE
therefore reduce the capacity available to the non-PIE choice problem.

In the deployed mechanism, PIE positions are allocated
according to the reconstructed PIE priorities and routing
procedures. The resulting assignments may include students
without PIE eligibility. SA-DA retains this procedure: PIE-eligible
students are selected first in PIE-priority order, followed by other
students in the order of the PIE subschool, and unused PIE capacity is
released to open competition.

The statutory basis is as follows. The successive priorities are established by Article~7~ter of DFL No.~2 of 1998, introduced by Law No.~20,845, and Decree No.~152 of 2016, particularly Articles~27 and~36. Article~37 specifies the disadvantaged-share calculation and Article~41 addresses sibling priority. The decree also regulates special admission procedures for students with permanent special educational needs.

The Ministry of Education reports a program-specific capacity for the
disadvantaged reserve, which we take as $r_s^{\dis}$. The
academic-excellence reserve changed form as the law was phased in. Under the
transitional regime, authorized schools could continue to use their
pre-existing special-admission procedure (\emph{Alta Exigencia}), with the
share of vacancies subject to it declining from 85 to 70 to 50 to 30 percent
over successive implementation years. Under the permanent regime, authorized
schools reserve 30 percent of vacancies for students satisfying the
prescribed academic-performance criterion. No program in any cycle carries
positive capacity in both forms. At each program we treat the active form as
the academic-excellence category and take its capacity as $r_s^{\exc}$.

The Ministry publishes public regular-stage microdata for the 2016--2025
application cycles \citep{mineduc_sae_data}. The 2016 Magallanes pilot
predates the reserve-era schema, carries no reserve-specific capacity fields,
and includes a nearest-school component based on perturbed coordinates; we
therefore report it for coverage but do not replicate it. For each of the nine
reserve-era cycles from 2017 through 2025, we reconstruct the deployed
mechanism, including grade sequencing, priority tiers, sibling and family
processing, incumbent guarantees, and the category sub-schools used in the
implementation of \citet{correa2022}. The data contain the realized lottery
information used by the mechanism. We combine the recorded priority fields
and lottery ranks with the tier and tie-breaking rules below; the
year-specific normalization conventions are stated in
\autoref{supp:routing-evidence}.

\autoref{tab:replication} reports the cycle-level replication results. Across
3{,}795{,}168 applicant--cycle observations, we exactly reproduce assignment
status, assigned school, and assigned course for every applicant.

In every counterfactual, submitted applications, published capacities,
substantive reserve eligibility, primitive priority criteria, realized lottery
ranks, and family links are held fixed. Grade sequencing, sibling and family
processing, and secured-enrollment rules are unchanged. Dynamic sibling
priority and the operational ordering of applications are recomputed using
assignments generated in higher grades within that counterfactual. For a returning student, the protected current program is
recorded as the outcome whenever the secured-enrollment guarantee is used, so
no returning student can be assigned below her protected program.

The regular-seat priority in the reconstruction orders returning students
first, followed by students with an already enrolled sibling, students
with a sibling newly assigned in a higher grade, children of staff, former
students, and other students. Lottery draws resolve ties within these tiers,
with the applicant identifier as a final tie-breaker. Reserve subschools use
their category-specific tiers. Within a priority tier, academically eligible
students in the transitional regime are ordered by the school-specific
\texttt{orden\_alta\_exigencia\_transicion} field, followed by lottery and
applicant identifier. The permanent academic reserve uses the lottery and
applicant identifier within each tier. In particular, returning students and both
sibling tiers precede the disadvantaged eligibility tier at the disadvantaged
subschool, and siblings are ordered in the same way at the disadvantaged
subschool and at regular seats. Returning students rank first for every seat
type at their current program. A non-eligible occupant of a reserve subschool
is treated as an open-seat assignee when measuring reserve utilization;
verifying priority-respect for such an assignment additionally requires the
open-seat priority comparison.

Within-program proposal orders also depend on applicant type. For
non-returning students, PIE comes first when active, followed by the academic
subschool. Disadvantaged students then propose to the disadvantaged subschool
before regular seats. Among non-disadvantaged students, siblings propose to
regular seats first, whereas non-siblings propose to the disadvantaged
subschool first. This last order permits non-disadvantaged students to
occupy disadvantaged-subschool capacity. At their current program, returning
students place disadvantage before the academic subschool; regular seats
precede disadvantage unless the returning student is disadvantaged. PIE is
first for returning students with recorded prior PIE status and last for other
returning students.

\begin{table}[H]
\centering
\caption{Replication of the deployed SAE, 2017--2025 cycles}
\label{tab:replication}
\small
\begin{tabular}{lrrr}
\hline
Cycle & Applicants & Assigned & Exact matches \\
\hline
2017 & 76{,}821  & 70{,}145  & 76{,}821  \\
2018 & 274{,}990 & 250{,}469 & 274{,}990 \\
2019 & 483{,}070 & 432{,}861 & 483{,}070 \\
2020 & 454{,}415 & 411{,}798 & 454{,}415 \\
2021 & 461{,}223 & 421{,}960 & 461{,}223 \\
2022 & 570{,}891 & 520{,}420 & 570{,}891 \\
2023 & 536{,}353 & 492{,}888 & 536{,}353 \\
2024 & 473{,}482 & 439{,}038 & 473{,}482 \\
2025 & 463{,}923 & 430{,}712 & 463{,}923 \\
\hline
Total & 3{,}795{,}168 & 3{,}470{,}291 & 3{,}795{,}168 \\
\hline
\end{tabular}
\begin{minipage}{0.82\textwidth}
\footnotesize\emph{Notes:} An exact match means that the reconstruction
matches the official assignment status, assigned school, and assigned course
for the applicant--cycle observation, using the initial destination fields
in the public result file. It does not compare reserve-assignment labels,
which are not reported in that file, or assignments after applicants respond.
\end{minipage}
\end{table}

\subsubsection{Observed outcomes and reconstructed routing}
\label{supp:routing-evidence}

The public inputs are the program offers (A1), applicant attributes (B1),
applications (C1), initial assignment results (D1), and family links (F1).
The D1 benchmark uses \texttt{rbd\_admitido} and
\texttt{cod\_curso\_admitido}; subsequent response and destination fields are
not part of this comparison. Inspection of all nine D1 schemas confirms that
none contains a reserve-assignment label. The reconstruction uses A1, B1, C1,
and F1 to generate both program and subschool assignments. Higher-grade
assignments from the same run generate the sibling-priority updates and
family-block promotions. Within each promoted or unpromoted group, the
published application order is preserved. Same-grade links do not generate
these higher-grade updates, and family-block reordering is disabled in 2017.

We retain applications marked \texttt{agregada\_por\_continuidad}, which the
system adds to protect continuity of enrollment. All 11{,}032 missing
\texttt{loteria\_original} values in the nine cycles occur on such returning
students' records in 2017. The reconstruction uses the applicant identifier
as a deterministic key when the lottery value is missing, and as the final
tie-breaker otherwise. It generates no new lottery draws.

PIE candidacy normally uses a nonmissing school-specific \texttt{orden\_pie}.
The separate \texttt{es\_pie} field records prior PIE status at the student's
current establishment. In 2017, when \texttt{orden\_pie} is nonmissing for
every application to a program, the reconstruction instead uses
\texttt{es\_pie}=1 to determine PIE candidacy there. This is an explicit
normalization convention; it changes the candidacy indicator on 134
application records relative to using nonmissing \texttt{orden\_pie} alone,
131 of them at programs with positive PIE capacity. The priority tiers and
within-program proposal orders above are then applied to these normalized
records.

An eligible holder of a focal subschool place is counted in that reserve;
an ineligible holder is counted as an open-seat assignee and checked against
the open-seat priority. Thus the physical subschool label and the
analytical reserve designation need not coincide.

Agreement on program assignments neither establishes uniqueness of the
reconstructed subschool routing nor identifies the normalization conventions.
Accordingly, claims concerning realized utilization, MRU failures, and
priority-respect of the deployed routing are conditional on this
reconstruction. The replication documentation provides the field-to-rule
mapping, the complete proposal-order table, and the checks for each cycle.

The reconstruction also recovers the preference list
$P_i^{\mathrm{op}}$ used for each student in the deployed
mechanism. Where family-block processing applies, the system
promotes programs at schools assigned to higher-grade members
of the student's declared block, preserving submitted order
within the promoted and remaining groups. These are the
preference lists processed by the reconstructed DA.
At an overlapping-reserve program, an excluded student
proposed to and was rejected by program $s$ exactly when
$s\mathrel{P_i^{\mathrm{op}}}\mu^{\mathrm{off}}(i)$.
We use these lists for deployed-outcome diagnostics in
\bodysecref{sec:res_uti_deployed}, and submitted rankings
for assignment comparisons in
\bodysecref{sec:sibling-adjusted-empirical}.

Empirical rank comparisons use the published
\texttt{preferencia\_postulante} field, retaining any current-program
application added by the system for continuity. Thus the recorded application
order, before family-block reordering, includes these continuity entries.

For the reserve-order analysis, groups use two applicant-level indicators.
Disadvantaged status is the recorded \texttt{prioritario} field. The academic
indicator is one if the student has at least one submitted application to an
overlapping-reserve program, carrying both an active academic-excellence
reserve and a disadvantaged reserve, at which she is eligible for the
academic-excellence reserve. Transitional eligibility uses a nonmissing
\texttt{orden\_alta\_exigencia\_transicion} field, whereas permanent-regime
eligibility uses \texttt{alto\_rendimiento}. The academic indicator need not
refer to either program assigned in the comparison. 
We use this profile definition in the reserve-order
comparison in \bodysecref{sec:sibling-adjusted-empirical}.

\subsection{Annual Exposure and Eligibility Overlap}
\label{supp:chile-exposure}

We call a program an \textbf{overlapping-reserve program} if it carries positive
capacity in both reserves, and a student \textbf{exposed} if her submitted list contains at least one such program. The overlap is potential: eligibility
for the two reserves can overlap, but a particular program's applicant pool
need not contain anyone eligible for both. \autoref{tab:exposure} therefore
reports exposure and realized eligibility overlap separately.

Exposure peaks in 2019 at 35 programs in 25 schools. From 2023 onward,
only three programs at one school remain, coinciding with the end of the
transitional academic regime. Through 2022, transitional capacity accounts
for 13{,}207 of the 13{,}343 academic places at overlapping-reserve programs;
thereafter all such capacity belongs to the permanent regime. The table
also documents the accompanying decline in eligibility overlap.

Distinct criteria need not produce conflicting rankings on a particular
applicant pool. Because each reserve's priority matters only among students
eligible for it, the realized disadvantaged and academic-excellence orders at
a program can be embedded in a single ranking exactly when they agree on the
students eligible for both. This happens in only 6 of the 127
overlapping-reserve program--cycle observations, at three distinct schools,
and at one of these schools the agreement is vacuous because no student is
eligible for both reserves. Thus, wherever eligibility overlap is nonvacuous,
the two criteria almost always conflict.

\begin{table}[H]
\centering
\caption{Exposure and eligibility overlap, 2017--2025 cycles}
\label{tab:exposure}
\small
\setlength{\tabcolsep}{5pt}
\begin{tabular}{lrrrrr}
\hline
Cycle & Programs & Exposed & Eligible $\geq 1$ & Eligible both & Both / $\geq 1$ \\
\hline
2017 & 6  & 3{,}871  & 3{,}243  & 1{,}170 & 36.1\% \\
2018 & 15 & 8{,}056  & 7{,}156  & 2{,}942 & 41.1\% \\
2019 & 35 & 16{,}941 & 14{,}293 & 5{,}042 & 35.3\% \\
2020 & 23 & 8{,}549  & 7{,}225  & 2{,}464 & 34.1\% \\
2021 & 23 & 8{,}383  & 7{,}177  & 2{,}671 & 37.2\% \\
2022 & 16 & 5{,}338  & 4{,}340  & 1{,}567 & 36.1\% \\
2023 & 3  & 1{,}601  & 1{,}227  & 261     & 21.3\% \\
2024 & 3  & 1{,}612  & 1{,}190  & 213     & 17.9\% \\
2025 & 3  & 1{,}506  & 1{,}127  & 174     & 15.4\% \\
\hline
Total & 127 & 55{,}857 & 46{,}978 & 16{,}504 & 35.1\% \\
\hline
\end{tabular}

\begin{minipage}{0.88\textwidth}
\footnotesize
\emph{Notes:}
``Programs'' counts overlapping-reserve programs, which carry positive
capacity in both the disadvantaged and active academic-excellence reserves;
the total counts program--cycle observations. ``Exposed'' counts students whose submitted list contains at least one such program; exposure does not
condition on the eventual assignment. Eligibility is
evaluated at the overlapping-reserve programs on the student's list.
``Eligible both'' counts a student if she is eligible for both reserves at
one or more of these programs. The final column divides ``Eligible both'' by
``Eligible $\geq 1$.'' Student counts are applicant--cycle observations.
\end{minipage}
\end{table}

\subsection{Additional Results on the Reserve Order under SA-DA}
\label{supp:sa-da-order}

\autoref{tab:sibling-adjusted-deployment} compares both versions of
SA-DA with deployment, and \autoref{tab:sibling-adjusted-order-national}
reports the effects of reversing the reserve order under SA-DA for each
admission cycle. The tables complement the aggregate and
eligibility-group comparisons in
\bodysecref{sec:sibling-adjusted-empirical}.

\begin{table}[htbp]
\centering
\caption{SA-DA relative to deployment}
\label{tab:sibling-adjusted-deployment}
\small
\setlength{\tabcolsep}{7pt}
\begin{tabular}{lrrr}
\toprule
& Deployed & Academic first & Disadvantaged first \\
\midrule
Different program & -- & 89 & 1{,}262 \\
Different school & -- & 88 & 1{,}255 \\
Higher-ranked program & -- & 47 & 635 \\
Lower-ranked program & -- & 42 & 627 \\
Unassigned $\to$ assigned & -- & 2 & 124 \\
Assigned $\to$ unassigned & -- & 5 & 119 \\
\midrule
Program--cycles with a shortfall & 25 & 0 & 0 \\
Shortfall (reserve positions) & 285 & 0 & 0 \\
\bottomrule
\end{tabular}

\begin{minipage}{0.94\textwidth}
\footnotesize
\emph{Notes:} Counts are applicant--cycle observations unless
otherwise indicated. Assignment changes are measured relative
to deployment. Preference comparisons use submitted program
rankings, with being unassigned ranked last. Shortfalls
are defined in \autoref{def:immediate-violation} and computed at
overlapping-reserve programs; by
\autoref{prop:sibling-adjusted-outcome}, SA-DA has no shortfall at
any program.
\end{minipage}
\end{table}

\begin{table}[htbp]
\centering
\caption{Effects of the reserve order under SA-DA}
\label{tab:sibling-adjusted-order-national}
\small
\setlength{\tabcolsep}{4pt}
\begin{tabular}{lrrrrr}
\toprule
& Different
& Different
& Prefer
& Prefer
& Same program, \\
Cycle
& program
& school
& academic first
& disadvantaged first
& different category \\
\midrule
2017 & 80  & 80  & 43  & 37  & 51  \\
2018 & 239 & 238 & 121 & 118 & 171 \\
2019 & 435 & 435 & 218 & 217 & 480 \\
2020 & 204 & 204 & 97  & 107 & 402 \\
2021 & 105 & 103 & 56  & 49  & 396 \\
2022 & 51  & 51  & 18  & 33  & 235 \\
2023 & 54  & 54  & 29  & 25  & 25  \\
2024 & 5   & 5   & 3   & 2   & 4   \\
2025 & 12  & 9   & 6   & 6   & 10  \\
\midrule
Total
& 1{,}185
& 1{,}179
& 591
& 594
& 1{,}774 \\
\bottomrule
\end{tabular}

\begin{minipage}{0.94\textwidth}
\footnotesize
\emph{Notes:} Counts are applicant--cycle observations.
Preference comparisons use submitted program rankings,
with being unassigned ranked last. The last column counts
applicants assigned to the same program under both orders
but through different admission categories; a disadvantaged holder of a
released position is counted in $\dis$.
\end{minipage}
\end{table}

\end{document}